\documentclass[12pt]{article}

\usepackage[utf8]{inputenc}
\usepackage{placeins}
\usepackage{graphicx}
\usepackage{float}
\usepackage{ragged2e}
\usepackage{amsthm}
\usepackage{amsmath}
\usepackage{mathtools}
\usepackage{amsfonts}
\usepackage{amssymb}
\usepackage[normalem]{ulem}
\usepackage{bbm}
\usepackage{bm}
\usepackage{tikz}
\usepackage{algpseudocode}
\usepackage{breqn}
\usepackage{booktabs}
\usepackage{threeparttable}
\usepackage{natbib}
\usepackage{nicefrac}
\usepackage[multiple]{footmisc}
\usepackage{hyperref}
\usepackage{enumitem}
\setenumerate{wide}
\usepackage{thmtools}
\usepackage{thm-restate}
\usepackage[nameinlink, capitalize]{cleveref}

\usepackage{cprotect}
\usepackage{etoc}
\usepackage{minitoc}
\usepackage{etaremune}

\usepackage[T1]{fontenc}    
\usepackage{booktabs}       
\usepackage{microtype}      
\floatstyle{ruled}
\usepackage{xcolor}

\usepackage{caption}
\usepackage{comment}
\usepackage{subcaption}

\usetikzlibrary{arrows,automata, calc, positioning,patterns}
\usepackage{adjustbox}
\usepackage{wrapfig}
\DeclareCaptionLabelSeparator{custom}{---}

\usepackage{multirow}
\usepackage{accents}

\usepackage[subtle]{savetrees}
\usepackage{apptools}

\usepackage{setspace}
\usepackage{etoolbox}
\usepackage{scalerel,stackengine}
\def\spacingset#1{\renewcommand{\baselinestretch}%
{#1}\small\normalsize} \spacingset{1}

\DeclareRobustCommand{\e}{\eps}

\DeclareRobustCommand{\E}{\mathbb{E}}

\DeclareFontFamily{U}{mathx}{}
\DeclareFontShape{U}{mathx}{m}{n}{<-> mathx10}{}
\DeclareSymbolFont{mathx}{U}{mathx}{m}{n}
\DeclareMathAccent{\widecheck}{0}{mathx}{"71}
\stackMath
\newcommand\reallywidehat[1]{%
\savestack{\tmpbox}{\stretchto{%
  \scaleto{%
    \scalerel*[\widthof{\ensuremath{#1}}]{\kern-.6pt\bigwedge\kern-.6pt}%
    {\rule[-\textheight/2]{1ex}{\textheight}}
  }{\textheight}%
}{0.5ex}}%
\stackon[1pt]{#1}{\tmpbox}%
}

\newtheorem{lemma}{Lemma}
\newtheorem{theorem}{Theorem}

\newtheorem{corollary}{Corollary}
\newtheorem{assumption}{Assumption}

\AtAppendix{\counterwithin{lemma}{section}}
\AtAppendix{\counterwithin{theorem}{section}}
\AtAppendix{\counterwithin{algorithm}{section}}
\AtAppendix{\counterwithin{corollary}{section}}
\AtAppendix{\counterwithin{assumption}{section}}
\AtAppendix{\counterwithin{proposition}{section}}

\theoremstyle{definition}

\newtheorem{definition}{Definition}
\newtheorem{example}{Example}
\newtheorem{remark}{Remark}

\AtAppendix{\counterwithin{definition}{section}}
\AtAppendix{\counterwithin{example}{section}}
\AtAppendix{\counterwithin{remark}{section}}



\begin{document}

\doparttoc 
\faketableofcontents 



\def\spacingset#1{\renewcommand{\baselinestretch}%
{#1}\small\normalsize} \spacingset{1}

\newcommand{\tightdisplays}{%
  \setlength{\abovedisplayskip}{2pt}%
  \setlength{\belowdisplayskip}{2pt}%
  \setlength{\abovedisplayshortskip}{2pt}%
  \setlength{\belowdisplayshortskip}{2pt}%
}

\makeatletter
\apptocmd{\normalsize}{\tightdisplays}{}{}
\apptocmd{\small}{\tightdisplays}{}{}
\apptocmd{\footnotesize}{\tightdisplays}{}{}
\makeatother


 \title{\bf Canonical correlation regression with noisy data}
  \author{
  Isaac Meza \\
  Department of Economics, Harvard University \\
  \and
  Rahul Singh\\
    Society of Fellows and Department of Economics, Harvard University
    }
     \date{December 2025}
  \maketitle

\bigskip
\begin{abstract}
We study instrumental variable regression in data rich environments. 
The goal is to estimate a linear model from many noisy covariates and many noisy instruments.
Our key assumption is that true covariates and true instruments are repetitive, though possibly different in nature; they each reflect a few underlying factors, however those underlying factors may be misaligned.
We analyze a family of estimators based on two stage least squares with spectral regularization: canonical correlations between covariates and instruments are learned in the first stage, which are used as regressors in the second stage. 
As a theoretical contribution, we derive upper and lower bounds on estimation error, proving optimality of the method with noisy data. 
As a practical contribution, we provide guidance on which types of spectral regularization to use in different regimes.
\end{abstract}

\noindent%

\vfill

\newpage
\spacingset{1.9} 

\section{Introduction and related work}

Linear instrumental variable regression is among the most widely used methods in economics, where it is increasingly implemented in data rich environments: settings where the economist has access to many noisy covariates and many noisy instruments. 
In macroeconomics and finance, covariates and instruments are often constructed from many noisy measurements of a few underlying factors.
In applied microeconomics, covariates and instruments constructed from texts, surveys, and transactions are high dimensional and contaminated by measurement error.

In data rich environments, a common practice is to use some version of principal component analysis to compress the information in the covariates and instruments before running two stage least squares regression. 
However, there is little practical guidance on how to choose among the competing spectral procedures. 
Moreover, there is little theoretical guidance on whether the recovered low-dimensional signal is strong enough---and, more subtly, well-aligned enough---to recover the parameter of interest despite the endogeneity and noise.
Our research question is how to optimally estimate the linear instrumental variable regression parameter in noisy, data rich environments where the signal of the covariates differs from the signal of the instruments.

Our key assumption is that the covariates and instruments are each repetitive in nature: they each are low rank, with a few factors providing the signal component to many noisy measurements. Crucially, we analyze not only regimes in which the covariate factors and instrument factors are well aligned, but also regimes in which they are poorly aligned. We interpret the extreme case in which the factors are perfectly aligned as a repeated measurement model. We interpret the other extreme case, in which the factors are perfectly misaligned, as an instrumental variable model with irrelevant instruments. We interpret the case of poorly aligned factors as an instrumental variable model with weak instruments. This paper provides a unified analysis across this continuum of regimes for a family of regularized estimators.

Our first contribution is to derive nonasymptotic upper bounds on the estimation error for a family of estimators that we refer to as canonical correlation regression. Different estimators in the family have different ``first stages'' and the same ``second stage.'' The first stage may be principal component analysis applied separately to the covariates and instruments; canonical correlation analysis applied jointly to the covariates and instruments; or any interpolation thereof. The second stage is regression upon the learned signal from the first stage. By providing a unified analysis for this family of estimators in the presence of noise, we characterize which version an economist should use, depending upon the regime. Our upper bound analysis culminates in a phase diagram, summarizing when to use which estimator.

Our second contribution is to derive a nonasymptotic lower bound on the estimation error in the presence of endogeneity and noise. Like the upper bound, the lower bound is valid across a continuum of regimes. By combining our upper bound and lower bound, we conclude that canonical correlation regression is sometimes optimal in noisy, data rich environments. A particular variation of the estimator is optimal in particular regimes. Intuitively, those regimes are ones in which the instrument is strong enough.

We contribute to a mature literature on ordinary least squares and two stage least squares with noisy factor models, by considering the realistic scenario in which the covariate factors and instrument factors are not perfectly aligned. 
The regression model is an instrumental variable regression model in which the covariates and instruments coincide. 
In such case, the estimator simplifies to principal component analysis of the covariates, followed by least squares. This estimator has been extensively studied in the context of noisy data under various names: factor augmented regression, diffusion indices, or principal component regression \citep{StockWatson2002,BaiNg2006,agarwalpcr,agarwal2024causalinferencecorrupteddata}.
To the best of our knowledge, previous analyses of two stage least squares with noisy data all place a strong assumption that we seek to relax: that the factors of the covariates and instruments are the same \citep{BaiNg2010,BaiWang2016}, and therefore that the instruments are essentially repeated measurements of the covariates. 
By contrast, in this work, we allow the factors of the covariates and the factors of the instruments to differ. We characterize the entire continuum from perfect alignment to perfect misalignment, which we view as the realistic continuum of settings faced by empirical economists.

We contribute to the mature literatures on high dimensional and nonparametric instrumental variable regression, with possibly weak instruments, by considering the realistic scenario in which the covariates and instruments are noisy.
The literature on many weak instruments is vast and growing \citep{Andrews2016CLC,Andrews2018TwoStep,MikushevaSun2022,LimWangZhangValidAR,LimWangZhangCLCWeak}, and some papers explicitly consider regularization \citep{CarrascoTchuente2016,DingGuoShiWang2025}, though all of these analyses appear to require clean data. 
For sparse and clean data, other works consider regularized approaches, under first stage restricted eigenvalue conditions 
\citep{belloni_chernozhukov_hansen,CARRASCO2012383,gautier2021highdimensionalinstrumentalvariablesregression}.
Finally, for nonparametric relationships in clean data, several works consider spectral regularization \citep{carrasco2007linear,darolles2011nonparametric,chen2018optimal,RenSunDaiSpectralIV,MeunierEtAlDemystify,van2025nonparametric}. The linear versions of these estimators often belong to the family of estimators we call canonical correlation regression. Our analysis is complementary to these various works because we consider measurement error.
Finally, we contribute to the literature on canonical correlation analysis by emphasizing its central role in economic data analysis. Canonical correlation analysis is a classical tool \citep{hotelling1936relations}; see e.g., \citet{bykhovskaya2025canonicalcorrelationanalysisreview,bykhovskaya2025highdimensionalcanonicalcorrelationanalysis} for recent reviews. We argue that it is often implicit in the analysis of instrumental variables, and is, in fact, a crucially important lens through which to view many estimators and many regimes. The asymptotic literature on canonical correlation analysis \citep{BenaychGeorgesNadakuditi2012,OnatskiMoreiraHallin2013,OnatskiMoreiraHallin2014,Dobriban2017, BaoHuPanZhou2019} clarifies the meaning of our finite-sample bounds. In particular, it helps us to interpret the weak-alignment settings in which sample canonical correlations are unreliable for instrumental variable estimation.

Section~\ref{sec:model} introduces the continuum of models.
Section~\ref{sec:algo} describes the family of estimators.
Section~\ref{sec:theory} derives the upper bounds on estimation error, and summarizes them with phase diagrams.
Section~\ref{sec:lower-bound} derives the minimax lower bound.
Section~\ref{sec:simulations} illustrates   our theory with simulations.
Section~\ref{sec:discussion} concludes.

\section{Many noisy covariates and instruments}\label{sec:model}

We formalize a data rich environment with many noisy covariates and instruments. We assume factor structure but allow a continuum of models in which the factors of covariates and instruments may be well aligned or poorly aligned.

As running convention, for a matrix $A$, we denote its transpose by $A^{\top}$ and its pseudo-inverse by $A^{\dagger}$. We denote its operator norm by $\|A\|_2$, and its $r$th singular value by $\sigma_r(A)$. Finally, we denote the projection onto the column space of $A$ by $\mathrm{proj}_A$.
For a vector $v$, we denote its $\ell_2$ norm by $\|v\|_2$. 
For scalars $a$ and $b$, we write $a\lesssim b$ when $a\leq C b$ for some universal constant $C$.
Expectations are with respect to the randomness of the structural disturbance $\varepsilon^*$ introduced below: $\mathbb{E}(\cdot \mid \cdot)=\int (\cdot )\mathrm{d}\mathbb{P}(\varepsilon^* \mid \cdot)$. 

\subsection{Instrumental variable model}

Consider the linear instrumental variable regression model with a scalar outcome, high dimensional covariates, and high dimensional instruments. In matrix notation, we express these quantities as $Y\in\mathbb{R}^n$, $X\in \mathbb{R}^{n\times p}$, and $W\in\mathbb{R}^{n\times q}$, respectively. We posit that the outcome $Y$ is generated by the linear model
$$
Y =X\beta_0+ \varepsilon^*,
\quad 
\mathbb{E}(\varepsilon^* \mid W) = 0,
$$
where $\varepsilon^*\in\mathbb{R}^n$ is a structural disturbance. In this causal model, 
$\mathbb E(\varepsilon^* \mid X)\neq 0$ yet $\mathbb{E}(\varepsilon^* \mid W) = 0$; the covariates $X$ are endogenous, so we leverage the instruments $W$. The causal parameter of interest is the coefficient vector $\beta^*\in\mathbb R^{p}$, which we aim to estimate well. We place a weak regularity condition on this coefficient.

\begin{assumption}[Parameter class]\label{ass:bounded_beta} The target parameter $\beta^*$ is the minimal $\ell_2$ norm element of the equivalence class $\{\beta : X\beta = X\beta_0\}$. It satisfies $ \|\beta^*\|_2\le B$ for a fixed radius $B>0$. \end{assumption}

The crux of our problem is that the economist does not directly observe $(Y,X,W)$; instead, the economist observes $(Y,Z_X,Z_W)$, where $Z_X\in\mathbb{R}^{n\times p}$ are noisy measurements of the covariates, and $Z_W\in \mathbb{R}^{n\times q}$ are noisy measurements of the instruments. We define the perturbations as $H_X=Z_X-X$ and $H_W=Z_W-W$. Under these definitions, 
$$
Z_X=X+H_X,\quad Z_W=W+H_W,
$$
and we have not yet placed any independence assumptions on the noise. These noise terms can capture measurement errors, discretizations, or privacy mechanisms applied to the true covariates and instruments.

What independence conditions are required on the structural disturbance $\varepsilon^*$ and the noise terms $H_X$ and $H_W$? For the structural disturbance, we have already assumed validity of the instrument via $\mathbb{E}(\varepsilon^* \mid W) = 0$. Throughout the paper, we also maintain that the structural disturbance is independent across observations and has bounded variance.

\begin{assumption}[Structural disturbance]\label{ass:var_bounded}
Conditional on $(X,W,H_X,H_W)$, the structural disturbances $\varepsilon_1^*,\dots,\varepsilon_n^*$
are independent with $\mathbb{E}(\varepsilon_i^*\mid W_i)=0$ and $\mathbb{E}\{(\varepsilon_i^*)^2\mid X,W,H_X,H_W\}\le \bar\sigma^2$.
\end{assumption}

For the noises, we remain largely agnostic: the noise terms on different covariates, on different instruments, and on covariates and instruments, may be weakly correlated and heteroscedastic. 
Our general, non-asymptotic results only use the independence and randomness of the structural disturbance; in this sense, they are conditional upon the signal and noise $(X,W,H_X,H_W)$. 

An economist may place additional structure on the problem to make our general results concrete.
For example, the operator norms of the noise $H_X$ and $H_W$ appear in our upper bounds. If the economist assumes the noise is sub-Gaussian and independent across observations, then these operator norms are small with high probability, appealing to the randomness of the noise. When deriving the lower bounds, we consider Gaussian noise for simplicity.

\subsection{Main assumption: Concentrated signal, diffuse noise}

Our strongest assumption is that the covariates and instruments each have factor structure. Algebraically, we assume they are low rank. Geometrically, we assume that the covariates and the instruments each have a low dimensional approximating subspace, though they are high dimensional. Moreover, the dimension of the covariate subspace is less than or equal to the dimension of the instrument subspace, generalizing the standard ``first stage'' rank condition in instrumental variable analysis.

\begin{assumption}[Low effective dimension and enough effective instruments]
\label{ass:rank}
We assume that $\mathrm{rank}(X)=k\ll \min(n,p)$ and $\mathrm{rank}(W)=\ell\ll \min(n,q)$ with $\ell\ge k$.
\end{assumption}

Since the covariates and instruments are low rank, i.e., well approximated by low dimensional subspaces, we now  introduce notation to describe those subspaces. Denote the singular value decompositions of the covariates and instruments by
\[
X = U_* \Sigma_* V_*^\top,\quad W=\widetilde{U}_* \widetilde{\Sigma}_* \widetilde{V}_*^\top.
\]
The diagonal matrices $\Sigma_*\in \mathbb{R}^{k\times k}$ and $\widetilde{\Sigma}_* \in \mathbb{R}^{\ell \times \ell}$ contain the singular values. They are pre and post multiplied by matrices that contain left and right singular vectors, respectively. The spans of these singular vectors define the low dimensional approximating subspaces.

To focus on the subspaces, it is sometimes helpful to discard the scaling by the singular values. For any rank-$r$ matrix $A=U_A\Sigma_A V_A^\top$, define its ``whitened'' factorization
$\underline A = U_A V_A^\top$.
In particular, $$\underline X = U_* V_*^\top,\quad \underline W=\widetilde{U}_* \widetilde V_*^\top.$$

In our results, the continuum of models is indexed by the alignment between the covariate factors and the instrument factors. Geometrically, this is the alignment between the covariate subspace and the instrument subspace. Algebraically, it is captured by an object we call the subspace overlap matrix
\[
\widetilde{U}_*^\top U_* \in \mathbb R^{\ell\times k}.
\]
The singular values of $\widetilde{U}_*^\top U_*$ lie in $[0,1]$ and equal the cosines of the principal angles between
$\mathrm{span}(U_*)$ and $\mathrm{span}(\widetilde{U}_*)$. Equivalently, they are the canonical
correlations between $\underline X$ and $\underline W$.

\begin{assumption}[Relevant but possibly weak instruments]\label{ass:full-overlap}
We assume $\widetilde{U}_*^\top U_*$ is full rank.
\end{assumption}

The rank of $\widetilde{U}_*^\top U_*$ is the number of covariate factors that can be learned from the instrument factors. 
By assuming $\widetilde{U}_*^\top U_*$ is full rank, we ensure that all of the directions of the covariate factors can be detected by the directions of the instrument factors. If $\widetilde{U}_*^\top U_*$ were rank deficient, then some covariate factors would be lost; the instruments would be irrelevant to those covariate factors. Importantly, we allow the smallest singular value of $\widetilde{U}_*^\top U_*$ to be close to zero, meaning that some covariate factors may only have weak instruments. In this way, we consider a continuum of instrumental variable models.

So far, we have placed assumptions on the signal: the covariates and instruments have factor structure, and the covariate factors can be at least weakly detected by the instrument factors. Now, we place an assumption on the noise: it is dominated by the signal in a particular sense.

\begin{assumption}[Low noise to signal ratio]
\label{ass:small-noise}
    Assume that $\left\|H_X\right\|_2 \lesssim \sigma_{k}(X)$ and $\left\|H_W\right\|_2 \lesssim \sigma_{\ell}(W)$.
\end{assumption}

Assumption~\ref{ass:small-noise} essentially means that the signal is concentrated while the noise is diffuse. By a concentrated signal, we mean that the smallest nonzero singular values of $X$ and $W$ are large enough, similar to strong factor, pervasiveness, and incoherence assumptions \citep{bai2016econometric}. By diffuse noise, we mean that the largest singular values of $H_X$ and $H_W$ are small enough, which is the case when their tails are well behaved. For example, Assumption~\ref{ass:small-noise} is satisfied when $X$ and $W$ have balanced nonzero singular values while $H_X$ and $H_W$ are sub-Gaussian, weakly dependent across columns, and independent across rows.

\section{Canonical correlation regression}\label{sec:algo}

We describe the family of two stage least squares estimators with spectral regularization. Then we articulate our research question in more detail.

\subsection{Family of estimators}

A consequence of Assumption~\ref{ass:small-noise} is that spectral regularization will be able to separate signal from noise. In this paper, we study a family of spectrally regularized two stage least squares estimators, which differ in how they regularize in the first stage. We will prove that different forms of first stage regularization lead to different performance, depending on the alignment between covariate factors and instrument factors as well as the strength of those factors.

To state the family of estimators, we extend the singular value decomposition notation from the true covariates and instruments to their noisy measurements. For the noisy analogues, we drop the subscript $*$. For example, we write
$$
Z_X = U\Sigma V^\top,
\quad
Z_W = \widetilde{U} \widetilde{\Sigma} \widetilde{V}^\top.
$$
Let $U_k\Sigma_k V_k^\top$ denote the rank-$k$ truncated singular value decomposition of $Z_X$,
and $\widetilde{U}_\ell \widetilde{\Sigma}_\ell \widetilde{V}_\ell^\top$ the rank-$\ell$ truncated decomposition of
$Z_W$. 
 When convenient, we abbreviate
$$
\widehat{X}= U_k\Sigma_k V_k^\top,
\quad
\widehat{W} = \widetilde{U}_\ell \widetilde{\Sigma}_\ell \widetilde{V}_\ell^\top.
$$
Finally, recall the whitening notation 
$\underline{\widehat{X}}= U_k V_k^\top$ and $\underline{\widehat{W}} = \widetilde{U}_\ell \widetilde{V}_\ell^\top.$

The spirit of two stage least squares is to project the covariates onto the instruments in the first stage, then to project the outcomes onto these projections in the second stage. Let $\widecheck{P}\in \mathbb{R}^{n\times p}$ be a high level representation of the projected covariates from the first stage. Then the abstract formulation of the second stage is
$$
  \widehat{\beta}
  =
  \arg\min_{\beta}\|Y-\widecheck{P}\beta\|_2^2.
$$
Compactly, the solution can be written as $ \widehat{\beta} = \widecheck{P}^{\dagger} Y$.

As a warm up, we write the oracle 2SLS estimator that has access to clean data $(Y,X,W)$. The oracle first stage would be the projection of $X$ onto $W$.
\begin{example}[Oracle 2SLS]\label{ex:oracle}
The 2SLS estimator using the clean covariates is $\beta^* = \widecheck{P}^\dagger Y$, where 
$$\widecheck{P} = \mathrm{proj}_W X= (\widetilde{U}_* \widetilde{U}_*^\top) (U_*\Sigma_* V_*^\top) = \widetilde{U}_*(I)( \widetilde{U}_*^\top U_*)(\Sigma_*) V_*^\top.$$
\end{example}
Intuitively, the first factor $\widetilde{U}_*$ is the signal of $W$. The last factor $V_*^\top$ is the signal of $X$. In the middle, we have the subspace overlap matrix $\widetilde{U}_*^\top U_*$, i.e. the true canonical correlations. It is pre-multiplied by the left weight $I$ and post multiplied by the right weight $\Sigma_*$.

In practice, the economist actually observes $(Y,Z_X,Z_W)$. Some sensible estimators spectrally regularize $Z_X$ and $Z_W$ in the first stage, either separately or jointly, before conducting least squares in the second stage.
\begin{example}[PCA-2SLS]\label{ex:pca}
Consider a 2SLS estimator that, in the first stage, conducts PCA on $Z_X$ and $Z_W$ separately, and then projects the former onto the latter. This estimator is $\widecheck{P}^\dagger Y$, where
$$
\widecheck{P}=\mathrm{proj}_{\widehat{W}}\widehat{X}=(\widetilde{U}_\ell \widetilde{U}_\ell^\top)(U_k\Sigma_k V_k^\top)=(\widetilde{U}_\ell)(I) (\widetilde{U}_\ell^\top U_k)(\Sigma_k) V_k^\top.
$$
\end{example}
This example is clearly an empirical analogue of the previous example. Instead of the true canonical correlations $\widetilde{U}_*^\top U_*$, it contains the estimated canonical correlations $\widetilde{U}_\ell^\top U_k$ as the middle factor. These are multiplied by the same left weight $I$ and the analogous right weight $\Sigma_k$. 

Alternative first stage procedures amount to different left and right weights on the estimated canonical correlations. 

\begin{example}[Whiten-2SLS]\label{ex:whiten}
Consider a 2SLS estimator that, in the first stage, conducts PCA on $Z_X$ and $Z_W$ separately, whitens them, and projects the former onto the latter. This estimator is $\widecheck{P}^\dagger Y$ where
$$
\widecheck{P}=\mathrm{proj}_{\underline{\widehat{W}}}\underline{\widehat{X}}=(\widetilde{U}_\ell \widetilde{U}_\ell^\top)(U_k V_k^\top)=(\widetilde{U}_\ell)(I) (\widetilde{U}_\ell^\top U_k)(I)V_k^\top.
$$
\end{example}
The left weight remains $I$ but now the right weight becomes $I$ as well.

\begin{example}[CCA-2SLS]\label{ex:cca}
Consider a 2SLS estimator that, in the first stage, conducts CCA on $Z_X$ and $Z_W$ jointly,
by forming a whitened cross-moment and truncating the spectrum. Concretely, define the joint CCA as
$\widehat{\underline{\widehat W}^{\top}\underline{\widehat X}}$, 
where the outer $\widehat{\cdot}$ denotes the rank-$k$ truncation.
This estimator is $\widecheck{P}^\dagger Y$ where
\begin{align*}
\widecheck P
&=\widehat W\big(\widehat W^{\top}\widehat W\big)^{\dagger}\widehat{\underline{\widehat W}^{\top}\underline{\widehat X}}
=\big(\widetilde U_\ell\widetilde\Sigma_\ell\widetilde V_\ell^\top\big)
   \big(\widetilde V_\ell\widetilde\Sigma_\ell^2\widetilde V_\ell^\top\big)^{\dagger}
   \left(\reallywidehat{\widetilde{V}_\ell\widetilde{U}_\ell^\top U_kV_{k}^\top}\right) \\
&= \widetilde{U}_{\ell}\widetilde{\Sigma}_\ell^{-1}\widetilde{V}_\ell^\top \left(\widetilde{V}_\ell{\widetilde{U}_\ell^\top U_k}V_{k}^\top\right) 
=
(\widetilde U_\ell)(\widetilde\Sigma_\ell^{-1})\big(\widetilde U_\ell^\top U_k\big)(I)V_k^\top.
\end{align*}
\end{example}

In this final example, the left weight becomes $\widetilde{\Sigma}_\ell^{-1}$ while the right weight becomes $I$.

\begin{definition}[Canonical correlation regression]
A canonical correlation regression is an estimator of the form $\widecheck{P}^\dagger Y$ where 
$$
\widecheck{P}= \widetilde{U}_\ell A_L(\widetilde{U}_\ell^\top U_k)A_R  V_k^\top.
$$
Here, $A_L \in \mathbb{R}^{\ell \times \ell}$ and $A_R \in \mathbb{R}^{k \times k}$ are diagonal weighting matrices. They apply different weights to the empirical canonical correlations $\widetilde{U}_\ell^\top U_k$.
\end{definition}

The initial factor is $\widetilde{U}_\ell$, i.e. the left singular vectors of the noisy instruments, while the final factor is $V_k$, i.e. the right singular vectors of the noisy covariates. The middle factor $\widetilde{U}_\ell^\top U_k$ is the empirical canonical correlations. The matrices $A_L \in \mathbb{R}^{\ell \times \ell}$ and $A_R \in \mathbb{R}^{k \times k}$ are diagonal weighting matrices that vary across different estimators within the family. By differently weighting the canonical correlations, they trade off signal usage against conditioning. As a consequence, different choices can dominate in different signal and noise regimes.

\subsection{Research question in detail}

A growing empirical literature uses some variation of principal component analysis to denoise and compress high dimensional covariates and instruments. 
These procedures are attractive because they are
computationally simple and robust to a wide range of measurement errors, but they leave a fundamental question unanswered: how to optimally estimate the linear instrumental variable regression parameter in noisy, data rich environments. This question has several components.

First, many different estimators that combine spectral regularization and two stage least squares are possible. 
It is not clear how to compare them on a common footing, and how to select one as a user. 
Practical guidance is missing.

Second, the statistical limits of such procedures are incompletely understood. 
How does the magnitude of the measurement error, together with the strength and alignment of instruments, jointly determine the best achievable accuracy?  
Theoretical guidance is missing.

Third, results from high-dimensional canonical correlation analysis
suggest that sufficiently weak correlations can be spectrally indistinguishable from noise. When is it the case that some components of $\beta^*$ may be impossible to learn?

In what follows, we develop theory for instrumental variable estimation with noisy, high dimensional data. Our analysis unifies across alignment regimes and estimators. We establish upper and lower bounds that clarify
which directions of $\beta^*$ are estimable, and at what rates.

\section{Upper bound on estimation error}\label{sec:theory}

As our first contribution, we analyze the estimation error for the family of estimators we call canonical correlation regression. We consider noisy, data rich environments. Our nonasymptotic upper bound decomposes the mean square error into interpretable terms. The main
bias–variance trade-off is governed by the choice of weights $(A_L,A_R)$. 

Our
bound yields a ``phase diagram'' in terms of singular values and condition numbers of
$(X,W)$. In some  regimes, PCA-style denoising is preferred, in others whitening is
preferred, and in a wide intermediate region CCA and closely related choices of $(A_L,A_R)$
dominate. 

\subsection{Key quantities}

A few interpretable quantities will recur in our upper and lower bounds: noise to signal ratios, and condition numbers. The noise to signal ratios are
$$
\texttt{NSR}_X = \frac{\|H_X\|_2^2}{\sigma_k(X)^2},
\quad
\texttt{NSR}_W = \frac{\|H_W\|_2^2}{\sigma_{\ell}(W)^2}.
$$
Intuitively, they measure the diffusion of the noise relative to the concentration of the signal, as measured by the largest singular value of the former and the smallest nonzero singular value of the latter. More concentrated signal and more diffuse noise, reflected by smaller noise to signal ratios,  will lead to better rates.

The condition numbers are $\kappa(A)=\sigma_{\max}(A)/\sigma_{\min}(A)$ and, when needed,
the generalized condition number $\kappa(A,A')=\sigma_{\max}(A)/
\sigma_{\min}(A')$. We will see condition numbers for the weights $(A_L,A_R)$ as well as for the canonical correlations $(\widetilde{U}_\ell^\top U_k)$. Intuitively, the condition numbers quantify how imbalanced the spectra of these objects are. Better conditioning, reflected by smaller condition numbers, will lead to better rates.

Finally, to lighten notation, we abbreviate a key object in the definition of the canonical correlation regression: the weighted canonical correlations are $\Delta=A_L(\widetilde{U}_\ell^\top U_k)A_R \in\mathbb{R}^{\ell\times k}$, so the second stage regressors can be expressed as $\widecheck{P}= \widetilde{U}_\ell \Delta  V_k^\top.$ 
Denote $r=\operatorname{rank}(\Delta) \leq \min\{\ell,k\}$.
Different $\Delta$ are given by different choices of $(A_L,A_R)$. 
Clearly, $\Delta$ is an estimator of the corresponding oracle quantity in Example~\ref{ex:oracle}: $\Delta_*=I( \widetilde{U}_*^\top U_*)\Sigma_*$. 
Considering alternative weights $(A_L,A_R)$ introduces bias, but possibly reduces variance.

\subsection{Decomposing mean square error}

As an initial step, we decompose the mean square error $\|\widehat{\beta}-\beta^*\|_2^2$ into three terms. The decomposition reflects the weighted canonical correlations $\Delta$. Specifically, we decompose the space of covariates $\mathbb{R}^p$ space via three operators:
$$
\Pi_{\text {row }}=V_k \mathrm{proj}_{\Delta} V_k^{\top}, \quad \Pi_{\text {null }}=V_k\left(I_k-\mathrm{proj}_{\Delta}\right) V_k^{\top}, \quad \Pi_{\perp}=I-V_k V_k^{\top} .
$$

These operators decompose the covariates based on their estimated signal. The projection onto the estimated signal is $V_k V_k^{\top}$. Clearly the initial two operators add up to the projection $V_k V_k^{\top}$, while the third operator is the rejection $I-V_k V_k^{\top}$. We call the third operator $\Pi_{\perp}$

The initial two operators further decompose the estimated covariate signal $V_k V_k^{\top}$ in a way that reflects the estimator, i.e. in a way that reflects the choice of weighted canonical correlations $\Delta$. The first operator $\Pi_{\text {row }}$ subsets to the directions of the estimated covariate signal that are preserved by the first stage of canonical correlation regression. Equivalently, it subsets to the row space of $\widecheck P$. The second operator $\Pi_{\text {null }}$ is the remainder.

Our decomposition of mean square error, formally derived in Lemma~\ref{lem:decomposition}, is 
$$
\|\hat\beta-\beta^*\|_2^2 = \|\Pi_{\text{row}}(\hat\beta-\beta^*)\|_2^2 + \left\|\left(\Pi^*_{\text{null}}- \Pi_{\text{null}}\right)\beta^*\right\|_2^2 + \left\|\left(\Pi^*_{\perp}-\Pi_{\perp}\right)\beta^*\right\|_2^2,
$$
where  $\Pi^*_{\text{null}} = V_*(I-\mathrm{proj}_{\Delta_*})V_*^\top$ and $\Pi^*_\perp = I-V_*V_*^\top$.

We bound the first term in Theorem~\ref{thm:pub} below. We bound the second and third terms in Theorem~\ref{thm:2s-ub} below. The first term generally dominates the second and third.

These three terms have an intuitive interpretation. The first term is the error of estimating the  coefficient, restricted to the directions that canonical correlation regression can actually learn, namely the row space of $\widecheck P$. The second term is error from extracting the causally relevant signal from the covariates using the estimated signal from the instruments. The third term is the error from estimating the signal of the noisy covariates.

\subsection{First main result}

To begin, we bound the first term in the mean square error.

\begin{theorem}[Upper bound on projected mean square error]
\label{thm:pub}
Under Assumptions~\ref{ass:bounded_beta}--\ref{ass:small-noise}, and~\ref{ass:full-column-rank}, 
     \begin{align*}
        \E\|\Pi_{\text {row }}\left(\widehat{\beta}-\beta^*\right)\|_2^2&\lesssim \frac{\left\|\left(I-A_L\right)\right\|_2^2\sigma_{\max}(X)^2 + \left\|A_L\right\|_2^2\left(\texttt{NSR}_X + \texttt{NSR}_W\right) \sigma_{\max}(X)^2 + \left\|A_L\right\|_2^2\left\|\Sigma_*-A_R\right\|_2^2}{c_{\ell}^2 c_k^2 \sigma_{\min }(\widetilde{U}_\ell^\top U_k)^2 \sigma_{\min }\left(A_L\right)^2 \sigma_{\min }\left(A_R\right)^2}\\
        &\quad+\texttt{NSR}_X\left(\frac{\kappa(\widetilde{U}_\ell^\top U_k)^2\kappa(A_L)^2\kappa(A_R)^2}{c_\ell^2c_k^2}\right)\\
        &\quad + \frac{\bar\sigma^2 \operatorname{rank}(\Delta)}{c_\ell^2c_k^2\sigma_{\min }(\widetilde{U}_\ell^\top U_k)^2\sigma_{\min}(A_L)^2\sigma_{\min}(A_R)^2}
    \end{align*}   
    for some constants $c_\ell, c_k$ bounded away from zero.
\end{theorem}

We interpret the upper bound as having three components: a weighting bias, a conditioning bias, and a variance.

The weighting bias is the first term. Some objects that feature prominently are $\left\|\left(I-A_L\right)\right\|_2^2$, $\left\|A_L\right\|_2^2$, and $\left\|\Sigma_*-A_R\right\|_2^2$, which are biases from weights that deviate from $A_L=I$ and $A_R=\Sigma_*$, which appear in the oracle estimator's $\Delta_*$. This term is larger when the noise to signal ratios are worse, and when the canonical correlations are closer to zero. We may view the denominator as the effective strength of the first stage of canonical correlation regression.

The conditioning bias is the second term. Some objects that feature prominently are the condition numbers of the weights,  $\kappa(A_L)$ and $\kappa(A_R)$, and of the canonical correlations $\kappa(\widetilde{U}_\ell^\top U_k)$. This term is larger when these objects are poorly conditioned, which magnifies the covariate noise.

The variance is the third term. Similar to the weighting bias, the denominator is the effective strength of the first stage of canonical correlation regression. Now, the numerator is the effective rank of canonical correlation regression: $\mathrm{rank}(\Delta)$. This term is large when the effective first stage is weak, or when the complexity of the estimator is high.

To complete our analysis of mean square error, we now bound its remaining components.

\begin{theorem}[Upper bound on remaining mean square error]
\label{thm:2s-ub}
Under Assumptions~\ref{ass:bounded_beta},~\ref{ass:full-overlap}, and~\ref{ass:small-noise},
    \begin{align*}
     \left\|\left(\Pi^*_{\text{null}}- \Pi_{\text{null}}\right)\beta^*\right\|_2^2 + \left\|\left(\Pi^*_{\perp}-\Pi_{\perp}\right)\beta^*\right\|_2^2&\lesssim \left\{\frac{\kappa(W)\kappa(X)}{\sigma_k(\widetilde{U}_*^\top U_*)}\right\}^2\left(\texttt{NSR}_W + \texttt{NSR}_X\right)+\frac{\kappa(W, X)^2\|\Sigma_*-A_R\|_2^2}{\sigma_{\ell}(W)^2\sigma_k(\widetilde{U}_*^\top U_*)^2}.
    \end{align*} 
\end{theorem}

Corollary \ref{cor:projected-dominated} shows how the bound in Theorem~\ref{thm:2s-ub} is generally dominated by bound in Theorem \ref{thm:pub}.

\subsection{Practical guidance}\label{sec:comparison}

To conclude our discussion of the upper bound on estimation error, we provide practical guidance for economists choosing among the estimators within the family we call canonical correlation regression. We simplify the bounds in Theorems~\ref{thm:pub} and~\ref{thm:2s-ub} by focusing on leading regimes. Under these simplifications, we summarize which estimators perform best via phase diagrams.

To begin, we articulate a condition under which the bound in Theorem~\ref{thm:pub} dominates the bound in Theorem~\ref{thm:2s-ub}. Specifically, under this condition, Corollary~\ref{cor:projected-dominated} confirms that the weighting bias in Theorem~\ref{thm:pub} dominates the quantities in Theorem~\ref{thm:2s-ub}.

\begin{assumption}[Well conditioned weights]
    \label{ass:tilting-spectrum}
   The weights satisfy 
    $\frac{\kappa(A_L)}{\sigma_{\min}(A_R)}=O\left\{\frac{\kappa(W)}{\sigma_{\min}(X)}\right\}.$
\end{assumption}

Having established that Theorem~\ref{thm:pub} dominates Theorem~\ref{thm:2s-ub}, we now establish which terms dominate within Theorem~\ref{thm:pub}. Under the following condition, the weighting bias dominates the conditioning bias. The condition essentially means that the covariates have a low enough noise to signal ratio.

\begin{assumption}[Low noise to signal]
    \label{ass:dominance-noise-A}
    The noise to signal ratios satisfy
    $\frac{\sigma_{\max }\left(A_R\right)}{\sigma_{\max }(X)} \lesssim \sqrt{1+\frac{\texttt{NSR}_W}{\texttt{NSR}_X}}.$
\end{assumption}

Finally, to simplify the diagrams, we assume the weighting bias is never too extreme.

\begin{assumption}[Limited weighting bias]
\label{ass:tilting-constant}
Within the weighting bias, 
$\left\|I-A_L\right\|_2^2 \sigma_{\max }(X)^2+\left\|A_L\right\|_2^2\left\|\Sigma_*-A_R\right\|_2^2$
    is at most constant order.
\end{assumption}

Under these simplifications, there are two remaining scenarios: either the weighting bias in Theorem~\ref{thm:pub} dominates, or the variance in Theorem~\ref{thm:pub} dominates. 

\begin{corollary}[Bias versus variance]
\label{cor:bias-var-dominated}
    Suppose the conditions of Theorems~\ref{thm:pub} and~\ref{thm:2s-ub} hold, as well as Assumptions~\ref{ass:tilting-spectrum}--\ref{ass:tilting-constant}. Then the bias dominated region is given by
$\texttt{NSR}_X+\texttt{NSR}_W\gtrsim \frac{\bar{\sigma}^2 r}{\kappa(X,W)^2}.$
        The variance-dominated region is given by
$\texttt{NSR}_X+\texttt{NSR}_W\lesssim \frac{\bar{\sigma}^2 r}{\kappa(X,W)^2}-\frac{1}{\kappa(X,W)^2}.$
\end{corollary}

These scenarios can be visualized in Figure~\ref{fig:dominates}. The horizontal axis quantifies how poorly the covariate factors and the instrument factors are aligned, via the cross conditioning $\kappa(X,W)$. A larger value on this axis means less signal alignment and less signal strength. The vertical axis quantifies how poorly behaved the noise to signal ratios are, via $\texttt{NSR}=\texttt{NSR}_X+\texttt{NSR}_W$. A larger value on this axis means more noise relative to the signal. Intuitively, bias dominates if the signals are poorly aligned and weak relative to the noise.

\begin{figure}
    \centering
    \begin{center}
\begin{tikzpicture}[scale=2.25]

  \draw[->] (0,0) -- (0,3) node[above] {$\texttt{NSR}$};
  \draw[->] (0,0) -- (4,0) node[right] {$\kappa(X,W)$};

  \draw[thick,domain=0.3:4,smooth,variable=\x] plot ({\x},{1/\x});
  \draw[thick,domain=0.75:4,smooth,variable=\x] plot ({\x},{2/\x});

  \node[right] at (0.175,0.95) {\textcolor{gray}{$\displaystyle \frac{\bar{\sigma}^2 r -1}{\kappa(X,W)^2}$}};
  \node[right] at (3,1.00) {\textcolor{gray}{$\displaystyle \frac{\bar{\sigma}^2 r}{\kappa(X,W)^2}$}};

  \fill[pattern=crosshatch,pattern color=gray!60,opacity=0.6]
    plot[domain=0.3:4,smooth,variable=\x] ({\x},{1/\x})
    -- plot[domain=4:0.75,smooth,variable=\x] ({\x},{2/\x})
    -- cycle;

  \node[anchor=west] at (0.2,0.25) {Variance-dominated};
  \node[anchor=west] at (1.5,1.6) {Bias-dominated};

\end{tikzpicture}
\end{center}
    \caption{Phase diagram of Corollary~\ref{cor:bias-var-dominated}}
    \label{fig:dominates}
\end{figure}
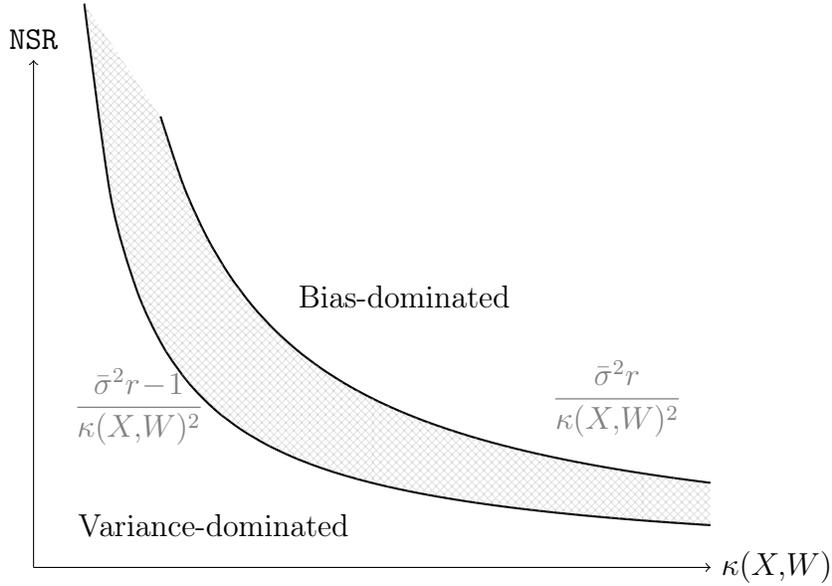

Finally, we present our practical contribution.
Within the bias-dominated region, we characterize which estimator performs best. This depends on $\sigma_{\min}(X)$ and $\kappa(W)$.
Within the variance-dominated region, we characterize which estimator performs best. This depends on $\sigma_{\min}(X)$ and $\sigma_{\max}(W)$.

\begin{corollary}[Estimator comparison]
\label{cor:estimator-dominance-region}
Suppose the conditions of Corollary~\ref{cor:bias-var-dominated} hold. When bias dominates, Figure~\ref{fig:bias} illustrates which estimator performs best among Examples~\ref{ex:pca},~\ref{ex:whiten}, and~\ref{ex:cca}. When variance dominates, Figure~\ref{fig:var} illustrates which estimator performs best among Examples~\ref{ex:pca},~\ref{ex:whiten}, and~\ref{ex:cca}. 
\end{corollary}

\begin{figure}
\begin{center}

\begin{tikzpicture}[scale=2.5]

  \def\e{0.1}

  \draw[->] (0,0) -- (3.2,0) node[right] {$\sigma_{\min}(X)$};
  \draw[->] (0,0) -- (0,3.0) node[above] {$\kappa(W)$};

  \draw (1,0) -- (1,-0.08) node[below] {$1$};
  \draw (0,1) -- (-0.08,1) node[left] {$1$};

  \draw[thick] (0,1) -- (1-\e,1);
  \draw[thick] (0,1-\e) -- (1,1-\e);
  \draw (3.025,1/3) -- (3.075,1/3);
  \draw (3.025,1/3-\e) -- (3.075,1/3-\e);
  \draw[thick] (3.05,1/3) -- (3.05,1/3-\e) node[right] {\footnotesize{$O(1)$}};
  
  \draw[dotted] (1,0) -- (1,1);
  \draw[thick] (1,1) -- (1,3);
  \draw[thick] (1-\e,1) -- (1-\e,3);


  \fill[pattern=crosshatch,pattern color=gray!60,opacity=0.6]
    (0,1-\e) -- (1,1-\e) -- (1,1) -- (0,1) -- cycle;

  \fill[pattern=crosshatch,pattern color=gray!60,opacity=0.6]
    (1-\e,1) -- (1,1) -- (1,3) -- (1-\e,3) -- cycle;

  \draw[thick,domain=1:3,smooth,variable=\x]
        plot ({\x},{1/\x});

  \draw[thick,domain=1:3,smooth,variable=\x]
        plot ({\x},{1/\x - \e});

  \fill[pattern=crosshatch,pattern color=gray!60,opacity=0.6]
    plot[domain=1:3,smooth,variable=\x] ({\x},{1/\x})
    -- plot[domain=3:1,smooth,variable=\x] ({\x},{1/\x - \e})
    -- cycle;

  \node[below right] at (2.4,0.85)
     {\textcolor{gray}{$\dfrac{1}{\sigma_{\min}(X)}$}};
     
  \node at (0.45,2.2) {Whitening};
  \node at (0.6,0.5) {CCA};
  \node at (2.2,2.1) {PCA};
\end{tikzpicture}    
\end{center}
    \caption{Which variation of canonical correlation regression to use when bias dominates}
    \label{fig:bias}
\end{figure}
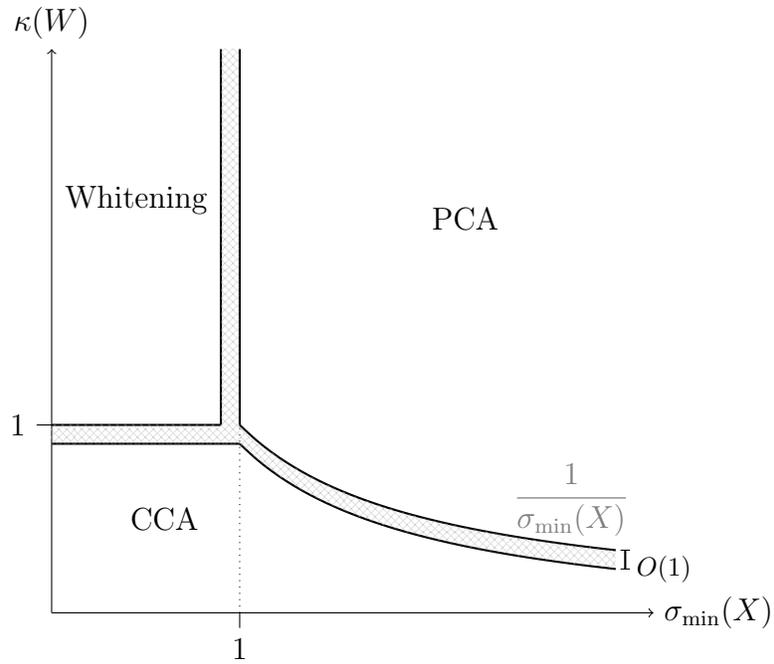

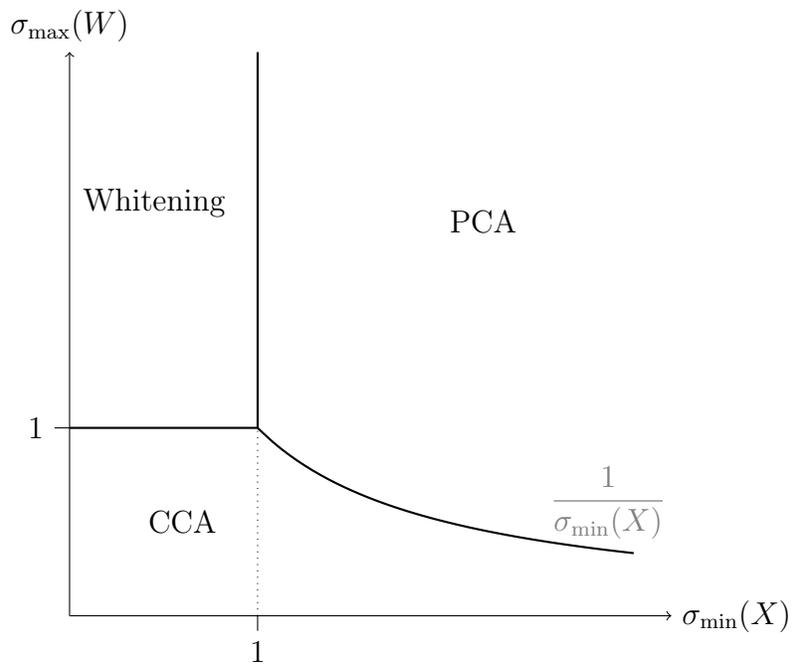
\begin{figure}
    \centering
\begin{center}
\begin{tikzpicture}[scale=2.5]

  \draw[->] (0,0) -- (3.2,0) node[right] {$\sigma_{\min}(X)$};
  \draw[->] (0,0) -- (0,3.0) node[above] {$\sigma_{\max}(W)$};

  \draw (1,0) -- (1,-0.08) node[below] {$1$};
  \draw (0,1) -- (-0.08,1) node[left] {$1$};

  \draw[thick] (0,1) -- (1,1);

  \draw[dotted] (1,0) -- (1,1);
  \draw[thick] (1,1) -- (1,3);

  \draw[thick,domain=1:3,smooth,variable=\x]
        plot ({\x},{1/\x});

  \node[below right] at (2.5,0.85)
     {\textcolor{gray}{$\dfrac{1}{\sigma_{\min}(X)}$}};

  \node at (0.45,2.2) {Whitening};
  \node at (0.6,0.5) {CCA};
  \node at (2.2,2.1) {PCA};

\end{tikzpicture}
\end{center}
    \caption{Which variation of canonical correlation regression to use when variance dominates}
    \label{fig:var}
\end{figure}

\section{Lower bound on estimation error}
\label{sec:lower-bound}

Having established an upper bound on the estimation error of canonical correlation regression, we now turn to a lower bound on the estimation error for any possible estimator that uses noisy covariates and instruments  $(Z_X,Z_W)$. We show that canonical correlation regression is essentially rate optimal among possible estimators, in certain regimes. Specifically, in regimes with balanced spectra, the variance term of our upper bound (Theorem~\ref{thm:pub}) matches the minimax lower bound (Theorem~\ref{thm:minimax-lb}) up to constants.

Our analysis also sheds light on the weak alignment regime, where canonical correlations are small due to poor alignment of the covariate factors and instrument factors, and therefore estimation breaks down. We interpret this negative result in light of classic work on weak instruments and recent work on high dimensional canonical correlation analysis.

\subsection{Second main result}

Throughout this section, we work conditionally on $(X,W)$. 
We view the structural disturbance and measurement error  $(\varepsilon^*,H_X,H_W)$ as Gaussian with known variances $(\sigma^2_{\varepsilon},\sigma^2_X,\sigma^2_W)$. Even with this additional structure, the lower bound will match our upper bound. 
The lower bound is obtained by a Fano-type
packing argument over $(X,W,\beta)$, where the different hypotheses are separated
in the parameter $\beta$ but induce probability distributions on
$(Y,Z_X,Z_W)$ that are difficult to distinguish.

\begin{theorem}[Minimax lower bound]
\label{thm:minimax-lb}
Under Assumption~\ref{ass:bounded_beta}, with $r_*=\operatorname{rank}(\Delta_*)$, 
\[
  \inf_{\widehat{\beta}}
  \sup_{\beta}
  \mathbb{E}\bigl[\|\widehat{\beta}-\beta\|_2^2\bigr]
  \gtrsim
  \max\left\{
    \frac{\sigma_\varepsilon^2  r_*  \sigma_{\min}(W)^2}
         {\sigma_{\max}(\widetilde{U}_*^\top U_*)^2 - \frac{r_* \sigma_X^2\sigma_{\min}(W)^2}{8}},
    \frac{\sigma_\varepsilon^2  r_*  \sigma_{\min}(W)^2}
         {\sigma_{\max}(\widetilde{U}_*^\top U_*)^2}
  \right\}.
\]
\end{theorem}

\begin{remark}[Simplification]
\label{rem:lb-low-rank-regime}
The two terms in the lower bound closely resemble each other, up to an extra subtraction in the denominator of the first term. This subtraction is negligible as long as we are in a low-rank regime:
$$
  r_*
  <
  \frac{8 \sigma_{\max}(\widetilde{U}_*^\top U_*)^2}
       {\sigma_X^2\sigma_{\min}(W)^2}
  =
  \frac{8}{\sigma_X^2}
  \kappa\!\bigl(\widetilde{U}_*^\top U_*,W\bigr).
$$
Then, the lower bound simplifies to one term: 
$
  \frac{\sigma_\varepsilon^2  r_*  \sigma_{\min}(W)^2}
         {\sigma_{\max}(\widetilde{U}_*^\top U_*)^2}
$.
\end{remark}

The lower bound fundamentally depends on the three quantities: the largest canonical correlation via $\sigma_{\max}(\widetilde{U}_*^\top U_*)^2$, the weakest instrument factor via $\sigma_{\min}(W)^2$, and the dimension of the oracle weighting via $r_*=\operatorname{rank}(\Delta_*)$. The fraction $\frac{\sigma_{\max}(\widetilde{U}_*^\top U_*)^2}{\sigma_{\min}(W)^2}$ may be viewed as the effective strength of the first stage. 
Theorem~\ref{thm:minimax-lb} shows that even with the best possible estimator,
the squared error must be at least $\sigma_\varepsilon^2 r_*$ divided by this fraction. The simplified lower bound closely resembles the variance term in Theorem~\ref{thm:pub}.

\begin{remark}[Optimality of canonical correlation regression]
\label{rem:lb-ub-match}
Consider the estimator in Example~\ref{ex:cca} with weights $A_L=\widetilde\Sigma_\ell^{-1}$ and $A_R=I$.
When the canonical correlations are balanced in the sense that
$\sigma_{\max}(\widetilde{U}_*^\top U_*)\asymp \sigma_{\min}(\widetilde{U}_*^\top U_*)$ and the instrument factors are balanced in the sense that $\sigma_{\max}(W)\asymp \sigma_{\min}(W)$, the variance term in
Theorem~\ref{thm:pub} matches the lower bound in Theorem~\ref{thm:minimax-lb} up to
constants. Therefore, in balanced spectral regimes where the variance term dominates the bias terms in
Theorem~\ref{thm:pub}, canonical correlation regression attains
the minimax rate (up to constants) among all estimators based on the noisy measurements $(Z_X,Z_W)$.
\end{remark}

Recall that the oracle estimator (Example~\ref{ex:oracle}) contained the true canonical correlations $\widetilde{U}_*^\top U_*$, and their weighted versions $\Delta_*$. We denote the number of these true canonical correlations by $r_*$. These oracle quantities appear in our mimimax lower bound (Theorem~\ref{thm:minimax-lb}).

Meanwhile, the feasible estimators (Example~\ref{ex:pca},~\ref{ex:whiten},~\ref{ex:cca}) must be computed from empirical canonical correlations $\widetilde{U}_\ell^\top U_k$, weighted as $\Delta$. We denote the number of the (weighted) empirical canonical correlations by $r$. The hyperparameters $(\ell,k)$, together with positive definite weights $(A_L,A_R)$, imply $r$. These empirical quantities appear in our upper bound of canonical correlation regression (Theorems~\ref{thm:pub} and~\ref{thm:2s-ub}).

Our upper bounds and lower bounds match when the empirical quantities match the oracle quantities. In practice, this amounts to ``oracle tuning'' of the hyperparameters $(\ell,k)$: they are correctly chosen, perhaps based on scree plots or some data driven procedure. As long as they are well chosen, and the factors are aligned well enough, the empirical canonical correlations will estimate the true canonical correlations well with high probability, implying that the upper and lower rates match.

By preserving the generality on non-oracle tuning in our upper bound analysis (Theorems~\ref{thm:pub} and~\ref{thm:2s-ub}), our upper rates for canonical correlation regression continue to describe its behavior under general forms of tuning.

\subsection{Weak instruments as weak alignment}

In this work, we view weak instruments as the setting where the covariate factors and instrument factors are weakly aligned, as measured by the canonical correlations $\widetilde{U}_*^\top U_*$. With weak enough canonical correlations, our lower bound explodes, and it becomes impossible to learn the instrumental variable regression coefficient $\beta^*$. This negative result echoes classic results on weak instruments \citep{StaigerStock1997}. For the remainder of this section, we relate this negative result to the more modern asymptotic theory on high dimensional canonical correlation regression \citep{BaikBenArousPeche2005,BenaychGeorgesNadakuditi2012,BaoHuPanZhou2019}. 

In a high-dimensional canonical correlation setting with $p,q,n\to\infty$ at proportions, the random-matrix literature shows that there is a critical canonical correlation $\rho_c$ such that: (i) if a population canonical correlation $\rho_j\leq \rho_c$ then the associated spike is spectrally indistinguishable from noise and the sample canonical directions are asymptotically orthogonal to the truth; (ii) if $\rho_j>\rho_c$ then the sample canonical direction has a nontrivial alignment with the population one, but its variability increases as $\rho_j\downarrow\rho_c$.

Our nonasymptotic, minimax lower bound is conditional upon the realized signal; in this sense, $(X,W)$ are fixed. In proportional asymptotics where $(X,W)$ are random, existing theory describes their spectral behavior and, in particular, separates weak-alignment directions (below $\rho_c$) from strong-alignment directions (above $\rho_c$).
Reading our lower bound through this lens suggests a refined understanding of the fundamental difficulty: there is an irreducible error associated with components of $\beta^*$ that lie in weakly-aligned canonical directions---which cannot be stably recovered when the corresponding effective canonical correlations do not separate from noise---and a variance inflation phenomenon for the estimable directions as the effective canonical correlations approach the critical level from above.

\section{Simulations}\label{sec:simulations}

We corroborate our theoretical results via simulations. The simulations document impressive finite-sample behavior of the canonical correlation regression estimators defined in
Section~\ref{sec:algo}. 

We compare four estimators. As a benchmark, we implement 2SLS with the noisy data, without any spectral regularization. Then, we implement three variations of canonical correlation regression: PCA-2SLS (Example~\ref{ex:pca}), Whiten-2SLS (Example~\ref{ex:whiten}), and CCA-2SLS (Example~\ref{ex:cca}). These spectrally regularized estimators use the same truncation hyperparameters $(k,\ell)$ for comparability.

We implement several synthetic low-rank data generating processes, each with many noisy covariates and instruments, and each with correlated measurement error in covariates and instruments. 
Figures~\ref{fig:high_dimensional_regime} and~\ref{fig:medium_dimensional_regime} consider high and moderate dimensional regimes, respectively. 
Within each figure, the subfigures consider regimes with very weak, weak, or strong alignment between the instrument factors and covariate factors. 
Within each subfigure, we consider different sample sizes and report mean square error. 
See Appendix~\ref{sec:details} for details on the data generating processes.

For each configuration, we run $250$ independent replications with fixed base seed and independent replication seeds, and report
$
\mathrm{MSE}(\hat\beta)=\frac{1}{p}\|\hat\beta-\beta^*\|_2^2,
$
together with pointwise $2.5\%$--$97.5\%$ quantile bands across replications.

In these data generating processes, the dominant contribution to error is the variance term, rather than bias from weighting or conditioning. 
This is exactly the regime in which our phase-diagram comparison is informative: performance differences are driven primarily by how each weighting shapes the  spectrum of $\Delta$. 

In this variance-dominated region, the phase diagram predicts that CCA-2SLS should outperform PCA-2SLS because the CCA weighting effectively stabilizes the contribution of poorly conditioned instrument directions. Equivalently, it avoids inflating the variance term by weighting towards the instruments. 

The empirical curves match our theoretical prediction: CCA-2SLS yields the smallest MSE, followed by Whiten-2SLS and then PCA-2SLS. Because we simulate matched power-law spectra for $X$ and $W$ (a ``balanced-spectrum'' design), our minimax lower bound matches the variance term up to constants in this regime. Therefore, improvements beyond the CCA-2SLS variance level are information-theoretically limited.

\begin{figure}[H]
  \centering
  \begin{subfigure}[t]{0.32\textwidth}
    \centering
    \includegraphics[width=\textwidth]{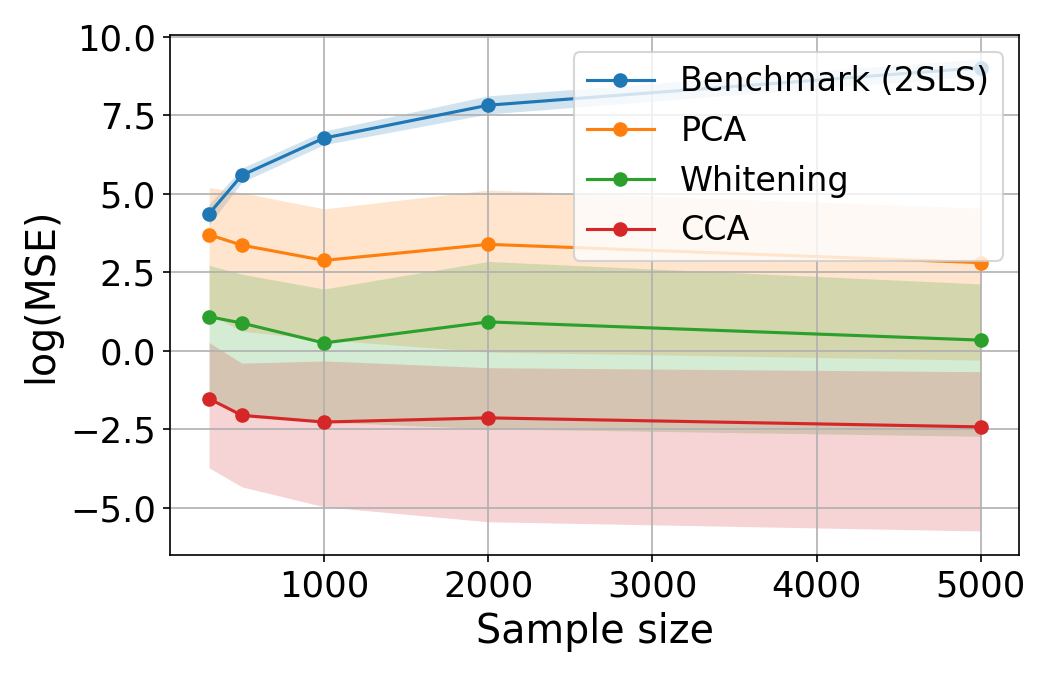}
    \caption{Very weak alignment}
    \label{fig:high_delta_0001}
  \end{subfigure}%
  \hfill
  \begin{subfigure}[t]{0.32\textwidth}
    \centering
    \includegraphics[width=\textwidth]{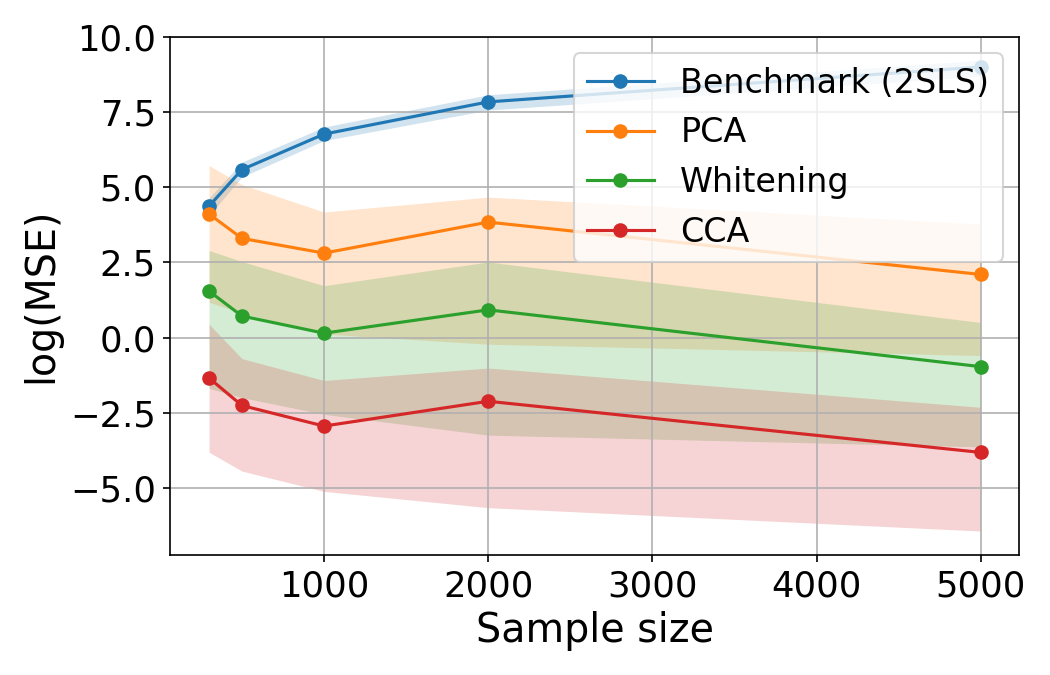}
    \caption{Weak alignment}
    \label{fig:high_delta_005}
  \end{subfigure}%
  \hfill
  \begin{subfigure}[t]{0.32\textwidth}
    \centering
    \includegraphics[width=\textwidth]{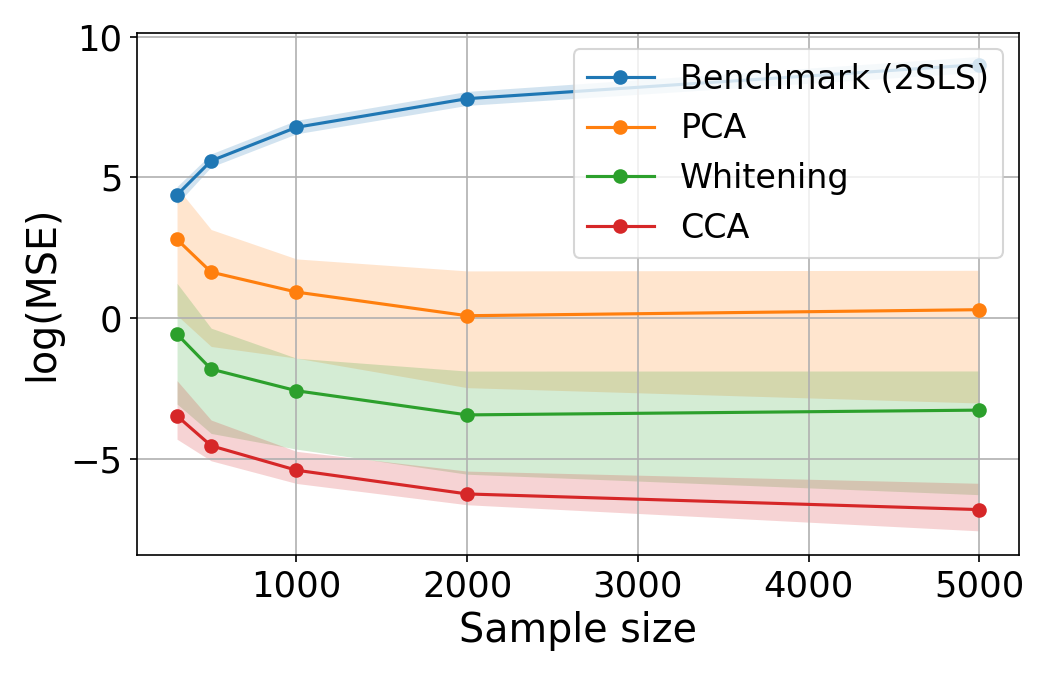}
    \caption{Strong alignment}
    \label{fig:high_delta_065}
  \end{subfigure}
  \caption{High dimensional regime}
  \label{fig:high_dimensional_regime}
\end{figure}

\begin{figure}[H]
  \centering
  \begin{subfigure}[t]{0.32\textwidth}
    \centering
    \includegraphics[width=\textwidth]{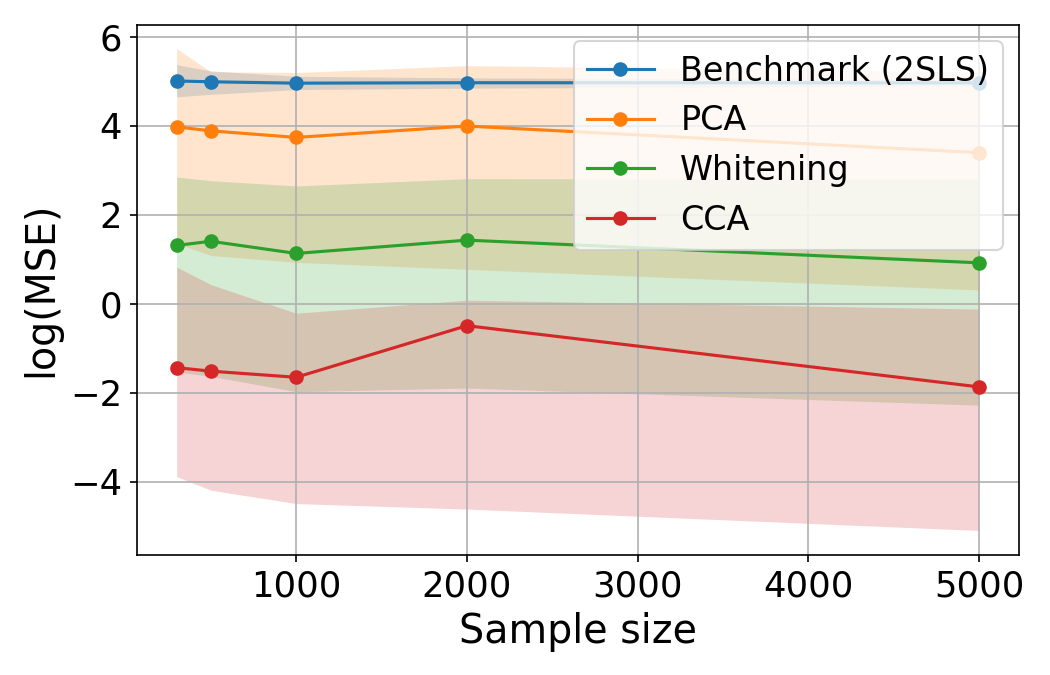}
    \caption{Very weak alignment}
    \label{fig:med_delta_0001}
  \end{subfigure}%
  \hfill
  \begin{subfigure}[t]{0.32\textwidth}
    \centering
    \includegraphics[width=\textwidth]{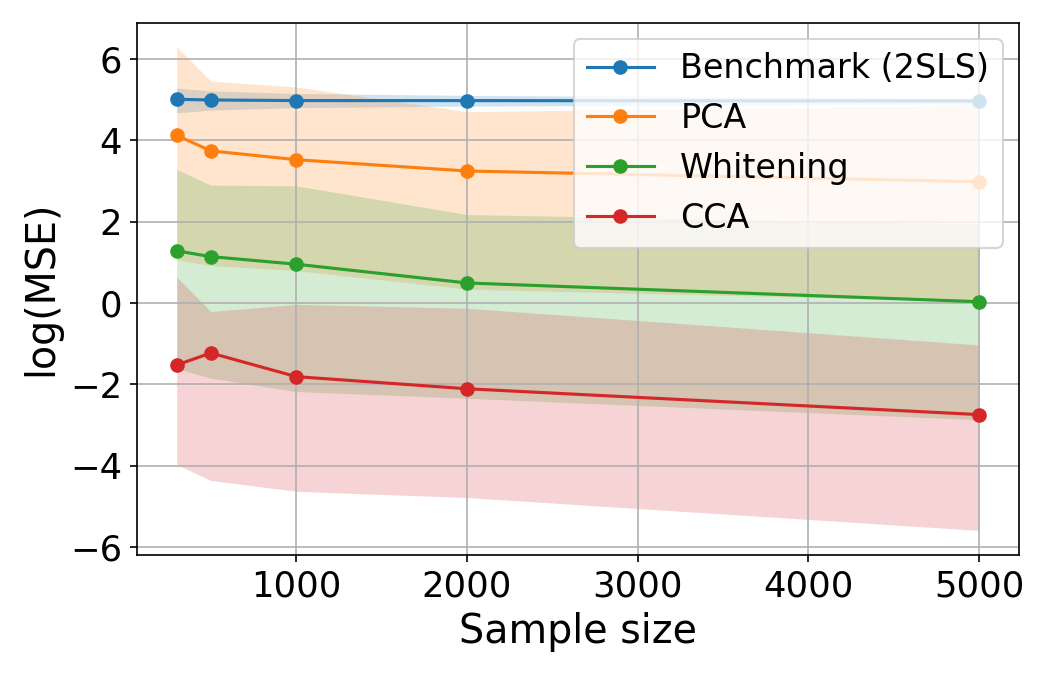}
    \caption{Weak alignment}
    \label{fig:med_delta_005}
  \end{subfigure}%
  \hfill
  \begin{subfigure}[t]{0.32\textwidth}
    \centering
    \includegraphics[width=\textwidth]{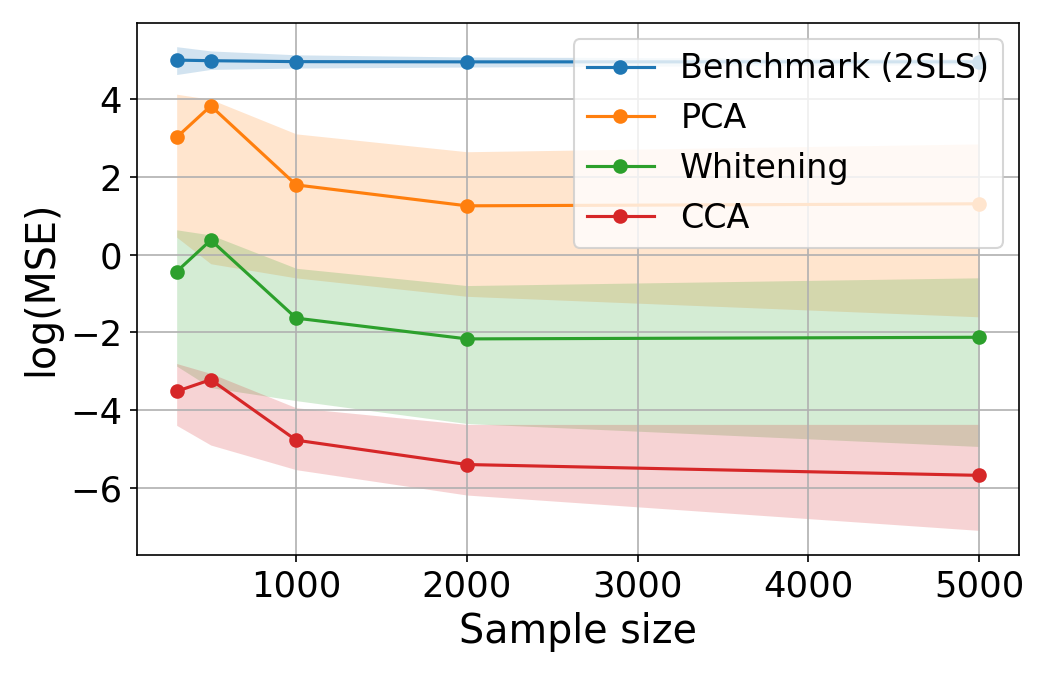}
    \caption{Strong alignment}
    \label{fig:med_delta_065}
  \end{subfigure}
  \caption{Moderate dimensional regime}
  \label{fig:medium_dimensional_regime}
\end{figure}
\section{Discussion}\label{sec:discussion}

This paper provides a unified analysis of low rank instrumental variable estimators with noisy, high dimensional data.
Our main contributions are sharp upper and lower bounds on estimation error when covariate factors and instrument factors differ from each other.
Theoretically, our results clarify how measurement error as well as alignment between the covariate factors and instrument factors govern the achievable accuracy.
Practically, our results distinguish the regimes when an economist should choose an estimator based on PCA, whitening, or CCA.
We conclude that canonical correlation regression is a natural and sometimes optimal method for exploiting 
low-rank structure in the presence of measurement error.

\bibliographystyle{hapalike_mod}
\spacingset{1}
{
\bibliography{0_refs}}
\spacingset{1.5}

\clearpage 
 
\appendix

\spacingset{1}
 \addcontentsline{toc}{section}{Appendix} 
\part{Appendix} 
\parttoc 
\spacingset{1.5}

 \section{Proofs}

\subsection{Additional notation}

To shorten the proofs, we define some additional notation. We write the subspace alignment matrices and their singular value decompositions (SVDs) as
$$
M_*=\widetilde{U}_*^\top U_*=Q_* S_* R_*^\top,\quad M=\widetilde U^\top U=QSR^\top, \quad \widehat{M}=\widetilde{U}_\ell^\top U_k = Q_\ell S R_k^\top.
$$
Formally, let $Q_\ell:=E_\ell^\top Q$ and $R_k:=F_k^\top R$ denote the row truncations that select the first $\ell$ and $k$ rows, respectively. Note that $Q_\ell$ and $R_k$ need not have orthonormal columns; they are row-truncations of orthonormal matrices for $M$.

For a matrix $A$, we write its row, column, and null spaces by $\operatorname{Row}(A)$, $\textsc{Col}(A)$, and $\textsc{Null}(A)$, respectively.

\subsection{Subspace decomposition}

\begin{lemma}[Subspace decomposition]
\label{lem:decomposition}
Let $Y \in \mathbb{R}^n$ and let $\widehat{\beta}=\widecheck{P}^{\dagger} Y$ be the minimal-norm least-squares solution to
$$
\min _{\beta \in \mathbb{R}^p}\|Y-\widecheck{P} \beta\|_2^2.
$$

\begin{enumerate}
    \item $\mathbb{R}^p$ decomposes orthogonally as
$$
\mathbb{R}^p=\underbrace{V_k \textsc{Col}\left(\Delta^{\top}\right)}_{\operatorname{Row}(\widecheck{P})} \oplus \underbrace{V_k \textsc{Null}(\Delta)}_{\text {uninformative within span(}V_k \text{) }} \oplus \underbrace{\textsc{span}\left(V_k\right)^{\perp}}_{\text {outside } X \text {-span }},
$$
with corresponding orthogonal projectors $\Pi_{\text {row }}, \Pi_{\text {null }}, \Pi_{\perp}$.
\item We have that $\widehat{\beta} \in \textsc{span}\{ V_k \textsc{Col}\left(\Delta^{\top}\right)\}$. Equivalently,
$$
\Pi_{\text {null }} \widehat{\beta}=0 \quad \text { and } \quad \Pi_{\perp} \widehat{\beta}=0 .
$$
\item These imply that
$$
\|\hat\beta-\beta^*\|_2^2 = \|\Pi_{\text{row}}(\hat\beta-\beta^*)\|_2^2 + \left\|\left(\Pi^*_{\text{null}}- \Pi_{\text{null}}\right)\beta^*\right\|_2^2 + \left\|\left(\Pi^*_{\perp}-\Pi_{\perp}\right)\beta^*\right\|_2^2.
$$
\end{enumerate}

\end{lemma}

\begin{proof}[Proof of Lemma \ref{lem:decomposition}]
We proceed in steps.
\begin{enumerate}
    \item Since $\widetilde{U}_{\ell}^{\top}$ has full row rank $\ell$,  $\textsc{Col}\left(\widecheck{P}^{\top}\right)=\textsc{Col}\left(V_k \Delta^{\top}\right)=V_k \textsc{Col}\left(\Delta^{\top}\right)$. Thus
\[
\operatorname{Row}(\widecheck{P})=\textsc{Col}\left(\widecheck{P}^{\top}\right)=V_k \textsc{Col}\left(\Delta^{\top}\right).
\]

Now, take a (thin) SVD $\Delta= L S T^{\top}$ with rank $r$. Then
$$
\mathrm{proj}_{\Delta}=\Delta^{\top}\left(\Delta \Delta^{\top}\right)^{\dagger} \Delta=T T^{\top},
$$
the orthogonal projector onto $\textsc{Col}\left(\Delta^{\top}\right)$. Therefore
$$
\Pi_{\mathrm{row}}=V_k \mathrm{proj}_{\Delta} V_k^{\top}=V_kTT^{\top} V_k^{\top}
$$
is the orthogonal projector in $\mathbb{R}^p$ onto $V_k \textsc{Col}\left(\Delta^{\top}\right)$.

Set $\Pi_{\text {null }}=V_k\left(I_k-\mathrm{proj}_{\Delta}\right) V_k^{\top}$ and $\Pi_{\perp}=I_p-V_k V_k^{\top}$. Symmetry and idempotence are immediate; mutual orthogonality follows from $V_k^{\top} V_{k, \perp}=0$ and $\mathrm{proj}_{\Delta}\left(I_k-\mathrm{proj}_{\Delta}\right)=0$. Their ranges are precisely the three subspaces stated above. Since $\mathrm{proj}_{\Delta}+\left(I_k-\mathrm{proj}_{\Delta}\right)=I_k$, we have $\Pi_{\text {row }}+\Pi_{\text {null }}+\Pi_{\perp}=I_p$. This proves the first result.

\item All least-squares solutions have the form
$$
\beta=\widecheck{P}^{\dagger} Y+\left(I_p-\widecheck{P}^{\dagger} \widecheck{P}\right) z, \quad z \in \mathbb{R}^p
$$
and $I_p-\widecheck{P}^{\dagger} \widecheck{P}$ is the orthogonal projector onto $\textsc{Null}(\widecheck{P})$. The minimal-norm solution is obtained with $z=0$, hence $\widehat{\beta}=\widecheck{P}^{\dagger} Y \in \operatorname{Row}\left(\widecheck{P}^{\top}\right)=V_k \textsc{Col}\left(\Delta^{\top}\right)$. Thus $\Pi_{\text {null }} \widehat{\beta}=\Pi_{\perp} \widehat{\beta}=0$. This proves the second result. 
\item The previous decomposition allows us to break down the estimation error. Let  $\Pi^*_{\text{null}} = V_*(I_k-\mathrm{proj}_{\Delta_*})V_*^\top$, $\Pi^*_\perp = I_p-V_*V_*^\top$. Observe that
\begin{align*}
    \hat\beta-\beta^* &= \Pi_{\text{row}}\left(\hat\beta-\beta^* \right) + \Pi_{\text{null}}\left(\hat\beta-\beta^* \right)  + \Pi_{\perp}\left(\hat\beta-\beta^* \right) \\
     &= \Pi_{\text{row}}\left(\hat\beta-\beta^* \right) + \Pi_{\text{null}}\left(-\beta^* \right)  + \Pi_{\perp}\left(-\beta^* \right) \\
    &=\Pi_{\text{row}}\left(\hat\beta-\beta^* \right) + \left(\Pi^*_{\text{null}}- \Pi_{\text{null}}\right)\beta^* + \left(\Pi^*_{\perp}-\Pi_{\perp}\right)\beta^*
\end{align*}
since by the result above, 
\[\Pi_{\text{null}}\hat\beta = 0,\quad\Pi_{\perp}\hat\beta=0,\quad\Pi^*_{\text{null}}\beta^* = 0,\quad\Pi^*_{\perp}\beta^*=0.\]
Finally, by the Pythagorean theorem
$$
\|\hat\beta-\beta^*\|_2^2 = \|\Pi_{\text{row}}(\hat\beta-\beta^*)\|_2^2 + \left\|\left(\Pi^*_{\text{null}}- \Pi_{\text{null}}\right)\beta^*\right\|_2^2 + \left\|\left(\Pi^*_{\perp}-\Pi_{\perp}\right)\beta^*\right\|_2^2.
$$

\end{enumerate}

\end{proof}

\subsection{Singular vector alignment}

From now on, denote the thin SVD of $\Delta = U_\Delta\Sigma_\Delta V_\Delta^\top$, then we can verify that $\Pi_{\text{row}} = V_kV_\Delta V_\Delta^\top V_k^\top$.

A fundamental quantity that appears repeatedly in our bounds is the operator norm of the deviation of the right singular vectors of $Z_X$ and left singular vectors of $Z_W$. To have proper bounds, we require the singular vectors to be properly aligned (Procrustes rotations).

Let
$$
Q_X \in \arg \min _{Q \in O(k)}\left\|V_*-V_k Q\right\|_2, \quad R_W \in \arg \min _{R \in O(\ell)}\left\|\widetilde{U}_*-\widetilde{U}_{\ell} R\right\|_2 .
$$
Define the aligned variables
$$
V_k^{\prime}:=V_k Q_X, \quad \widetilde{U}_{\ell}^{\prime}:=\widetilde{U}_{\ell} R_W, \quad \Delta^{\prime}:=R_W^{\top} \Delta Q_X .
$$
Then $\widetilde{U}_{\ell} \Delta V_k^{\top}=\left(\widetilde{U}_{\ell}^{\prime}\right) \Delta^{\prime}\left(V_k^{\prime}\right)^{\top}$, so the design operator $\widecheck{P}$ (and hence $\widecheck{P}^{\dagger}$ and $ \Pi_{\text {row }}$) is unchanged.
Write an SVD of $\Delta$ as $\Delta=U_{\Delta} \Sigma_{\Delta} V_{\Delta}^{\top}$.
Then the SVD of $\Delta^{\prime}$ is $\Delta^{\prime}=\left(R_W^{\top} U_{\Delta}\right) \Sigma_{\Delta}\left(Q_X^{\top} V_{\Delta}\right)^{\top}$ so $\Sigma_{\Delta^{\prime}}=\Sigma_{\Delta}, U_{\Delta^{\prime}}=R_W^{\top} U_{\Delta}, V_{\Delta^{\prime}}=Q_X^{\top} V_{\Delta}$. Finally, after alignment \cite[Lemma 1]{Cai_2018_aos}:
\begin{align*}
\left\|V_*-V_k^{\prime}\right\|_2& \leq \sqrt{2} \sin \Theta_{\max}\left(\textsc{span}(V_*),\textsc{span}(V_k)\right) \\
\left\|\widetilde{U}_*-\widetilde{U}_{\ell}^{\prime}\right\|_2& \leq \sqrt{2} \sin \Theta_{\max}\left(\textsc{span}(\widetilde{U}_*),\textsc{span}(\widetilde{U}_{\ell})\right). 
\end{align*}
Wedin’s $\sin\text{-}\Theta$ theorem upper bounds these quantities.

\begin{lemma}[Procrustes analysis]
    \label{lem:Vk_Ul}
   Under Assumption \ref{ass:small-noise}, and the optimal alignment from above, we have the upper bounds
    \begin{align*}
\left\|V_*-V_k^{\prime}\right\|_2\leq \sqrt{2} \frac{\left\|H_X\right\|_2}{\sigma_k(X)}, \quad\left\|\widetilde{U}_*-\widetilde{U}_{\ell}^{\prime}\right\|_2\leq \sqrt{2} \frac{\left\|H_W\right\|_2}{\sigma_{\ell}(W)}. 
\end{align*}
\end{lemma}
\begin{proof}
    By a direct application of Wedin's Theorem given the bound obtained above. 
\end{proof}

Henceforth, we work with these aligned representatives; without loss of generality, we may assume that we are already in the aligned gauge and we can upper bound the quantities according to the previous Lemma \ref{lem:Vk_Ul}.

\subsection{Projected coefficient consistency}
In this section we provide the proof of Theorem \ref{thm:pub}.

\begin{lemma}[Decomposition of projected mean square error]
Under Assumptions \ref{ass:rank}, and \ref{ass:var_bounded}, then
    \begin{align*}
        \E\left\|\Pi_{\text {row }}\left(\widehat{\beta}-\beta^*\right)\right\|_2^2&\lesssim \sigma_{\min}(\Delta)^{-2} \|\beta^*\|_2^2\left(\left\|\Delta_* -\Delta\right\|_2^2 +\left\|V_*^\top-V_k^{\top}\right\|_2^2\sigma_{\max}(\Delta)^2+\left\|\widetilde{U}_{\ell}-\widetilde{U}_*\right\|_2^2\sigma_{\max}(\Sigma_*)^2\right) \\
        &\qquad\qquad + \bar\sigma^2 \sum_{i=1}^{\operatorname{rank}(\Delta)} \frac{1}{\sigma_i\left(\Delta\right)^2}.
    \end{align*}
\end{lemma}

\begin{proof}
Since the estimator is given by $\widehat{\beta}=\widecheck{P}^{\dagger} Y$ and the true outcome model is $Y=X \beta^*+\varepsilon^*$, we can write:
\begin{align}
    \label{eq:vk_beta}
V_k^{\top} \widehat{\beta}=V_k^{\top} V_k \Delta^{\dagger} \widetilde{U}_{ell}^{\top} Y=\Delta^{\dagger} \widetilde{U}_{ell}^{\top}\left(X \beta^*+\varepsilon^*\right).
\end{align}

Using the SVD $\Delta=U_{\Delta} \Sigma_{\Delta} V_{\Delta}^{\top}$, note that
\[\Pi_{\mathrm{row}} = V_k V_{\Delta} V_{\Delta}^{\top} V_k^{\top} = V_k V_{\Delta} U_{\Delta}^{\top} U_{\Delta} V_{\Delta}^{\top} V_k^{\top}=L^{\top} L,\]
where we define $L:=U_{\Delta} V_{\Delta}^{\top} V_k^{\top}$. Consequently, 
$$
\left\|\Pi_{\text {row }}\left(\widehat{\beta}-\beta^*\right)\right\|_2=\left\|L\left(\widehat{\beta}-\beta^*\right)\right\|_2=\left\|U_{\Delta} V_{\Delta}^{\top} V_k^{\top}\left(\widehat{\beta}-\beta^*\right)\right\|_2 .
$$

Since $U_{\Delta}$ has orthonormal columns, then
\begin{align}
    \label{eq:proj_norm}
\left\|\Pi_{\text {row }}\left(\widehat{\beta}-\beta^*\right)\right\|_2\leq\left\|V_{\Delta}^{\top}\left(V_k^{\top} \widehat{\beta}-V_k^{\top} \beta^*\right)\right\|_2 .
\end{align}

From (\ref{eq:vk_beta}) and $\Delta^{\dagger}=V_{\Delta} \Sigma_{\Delta}^{-1} U_{\Delta}^{\top}$,
$$
V_{\Delta}^{\top} V_k^{\top} \widehat{\beta}=V_{\Delta}^{\top} \Delta^{\dagger} \widetilde{U}_{ell}^{\top}\left(X \beta^*+\varepsilon^*\right)=\Sigma_{\Delta}^{-1} U_{\Delta}^{\top} \widetilde{U}_{ell}^{\top}\left(X \beta^*+\varepsilon^*\right). 
$$

Plugging this into (\ref{eq:proj_norm}) yields
\begin{align*}
\left\|\Pi_{\text {row }}\left(\widehat{\beta}-\beta^*\right)\right\|_2&=\left\|\Sigma_{\Delta}^{-1} U_{\Delta}^{\top} \widetilde{U}_{\ell}^{\top} X \beta^*-V_{\Delta}^{\top} V_k^{\top} \beta^*+\Sigma_{\Delta}^{-1} U_{\Delta}^{\top} \widetilde{U}_{\ell}^{\top} \varepsilon^*\right\|_2 \\   &\leq\underbrace{\left\|\left(\Sigma_{\Delta}^{-1} U_{\Delta}^{\top} \widetilde{U}_{\ell}^{\top} X -V_{\Delta}^{\top} V_k^{\top}\right) \beta^*\right\|_2}_{\text{Bias}}+\underbrace{\left\|\Sigma_{\Delta}^{-1} U_{\Delta}^{\top} \widetilde{U}_{\ell}^{\top} \varepsilon^*\right\|_2}_{\text{Variance}}.
\end{align*}

We now focus on the bias term:
\begin{align*}
    \left\|\left(\Sigma_{\Delta}^{-1} U_{\Delta}^{\top} \widetilde{U}_{\ell}^{\top} X -V_{\Delta}^{\top} V_k^{\top}\right) \beta^*\right\|_2 &= \left\|\Sigma_{\Delta}^{-1} U_{\Delta}^{\top}\left( \widetilde{U}_{\ell}^{\top} X -\Delta V_k^{\top}\right) \beta^*\right\|_2 \\
    &\leq \sigma_{\min}(\Delta)^{-1}\left\|\widetilde{U}_{\ell}^{\top} X -\Delta V_k^{\top}\right\|_2\|\beta^*\|_2 \\
    &\leq \sigma_{\min}(\Delta)^{-1} \|\beta^*\|_2\left(\left\|\widetilde{U}_{*}^{\top} X -\Delta V_k^{\top}\right\|_2+\left\|(\widetilde{U}_{\ell}-\widetilde{U}_*)^\top X\right\|_2\right) \\
    &= \sigma_{\min}(\Delta)^{-1} \|\beta^*\|_2\left(\left\|\widetilde{U}_{*}^{\top} U_*\Sigma_*V_*^\top -\Delta V_k^{\top}\right\|_2+\left\|(\widetilde{U}_{\ell}-\widetilde{U}_*)^\top X\right\|_2\right) \\
    &= \sigma_{\min}(\Delta)^{-1} \|\beta^*\|_2\left(\left\|\Delta_*V_*^\top -\Delta V_k^{\top}\right\|_2+\left\|(\widetilde{U}_{\ell}-\widetilde{U}_*)^\top X\right\|_2\right) \\
    &= \sigma_{\min}(\Delta)^{-1} \|\beta^*\|_2\left(\left\|\Delta_*V_*^\top -\Delta V_k^{\top}\right\|_2+\left\|(\widetilde{U}_{\ell}-\widetilde{U}_*)^\top U_*\Sigma_*V_*^\top\right\|_2\right) \\
    &= \sigma_{\min}(\Delta)^{-1} \|\beta^*\|_2\left(\left\|\Delta_*V_*^\top -\Delta V_k^{\top}\right\|_2+\left\|\widetilde{U}_{\ell}^\top U_*\Sigma_*V_*^\top-\widetilde{U}_\ell^\top U_*\Sigma_*V_*^\top\right\|_2\right) \\
    &\leq \sigma_{\min}(\Delta)^{-1} \|\beta^*\|_2\left(\left\|\Delta_*V_*^\top -\Delta V_k^{\top}\right\|_2+\left\|\widetilde{U}_{\ell}-\widetilde{U}_*\right\|_2\sigma_{\max}(\Sigma_*)\right) \\
    &\leq \sigma_{\min}(\Delta)^{-1} \|\beta^*\|_2\left(\left\|\Delta_* -\Delta\right\|_2 +\left\|V_*^\top-V_k^{\top}\right\|_2\sigma_{\max}(\Delta)+\left\|\widetilde{U}_{\ell}-\widetilde{U}_*\right\|_2\sigma_{\max}(\Sigma_*)\right).
\end{align*}

The variance term is bounded as
\begin{align*}
\E\left[\left\| \Sigma_{\Delta}^{-1} U_{\Delta}^{\top} \widetilde{U}_{\ell}^{\top}  \varepsilon^* \right\|_2^2\right] &= \operatorname{Tr}\left( \Sigma_{\Delta}^{-1} U_{\Delta}^{\top} \widetilde{U}_{\ell}^{\top} \E[\varepsilon^*\varepsilon^{*\top} | X, Z_X, W, Z_W]  \widetilde{U}_{\ell}U_{\Delta} \Sigma_{\Delta}^{-1}  \right) \\
&\le \bar\sigma^2 \operatorname{Tr}\left(\Sigma_{\Delta}^{-1} U_{\Delta}^{\top} \widetilde{U}_{\ell}^{\top} \widetilde{U}_{\ell}U_{\Delta} \Sigma_{\Delta}^{-1}\right) \\
&= \bar\sigma^2 \left\|\Sigma_{\Delta}^{-1}  \right\|_{\mathrm{Fr}}^2 \\
&= \bar\sigma^2 \sum_{i=1}^{\operatorname{rank}(\Delta)} \frac{1}{\sigma_i\left(\Delta\right)^2}.
\end{align*}

\end{proof}

The error term stemming from the singular vectors are bounded using Lemma \ref{lem:Vk_Ul}. It remains to bound $\|\Delta_*-\Delta\|_2^2$ for a complete treatment of the bias term. 

\begin{lemma}[Weighting bias bound]
Consider the following \emph{estimator} for $\Delta_*$:
\[\Delta(A_L,A_R) = A_L(\widetilde{U}_\ell^\top U_k)A_R,\]
where $A_L \in \mathbb{R}^{\ell \times \ell}, A_R \in \mathbb{R}^{k \times k}$ diagonal and positive definite matrices. Under Assumption \ref{ass:small-noise} then,
\[\left\|\Delta_*-\Delta\left(A_L, A_R\right)\right\|_2\lesssim\left\|\left(I-A_L\right)\right\|_2\sigma_{\max}(\Sigma_*) + \left\|A_L\right\|_2\left(\frac{\|H_X\|_2}{\sigma_{k}(\Sigma_*)} + \frac{\|H_W\|_2}{\sigma_{\ell}(\widetilde{\Sigma}_*)}\right) \sigma_{\max}(\Sigma_*) + \left\|A_L\right\|_2\left\|\Sigma_*-A_R\right\|_2.\]
\end{lemma}
\begin{proof}
We write
\begin{align*}
\left\|\Delta_*-\Delta\left(A_L, A_R\right)\right\|_2 &=  \left\|M_* \Sigma_*-A_L M A_R\right\|_2 \\
&= \left\|\left(I-A_L\right) M_* \Sigma_* + A_L\left(M_*-M\right) \Sigma_* + A_L M\left(\Sigma_*-A_R\right)\right\|_2 \\
&\leq \left\|\left(I-A_L\right) M_* \Sigma_*\right\|_2 + \left\|A_L\left(M_*-M\right) \Sigma_*\right\|_2 + \left\|A_L M\left(\Sigma_*-A_R\right)\right\|_2 \\
&\leq \left\|\left(I-A_L\right)\right\|_2\sigma_{\max}(\Sigma_*) + \left\|A_L\right\|_2\left\|M_*-M\right\|_2 \sigma_{\max}(\Sigma_*) + \left\|A_L\right\|_2\left\|\Sigma_*-A_R\right\|_2 \\
&\leq \left\|\left(I-A_L\right)\right\|_2\sigma_{\max}(\Sigma_*) + \sqrt{2}\left\|A_L\right\|_2\left(\frac{\|H_X\|_2}{\sigma_{k}(\Sigma_*)} + \frac{\|H_W\|_2}{\sigma_{\ell}(\widetilde{\Sigma}_*)}\right) \sigma_{\max}(\Sigma_*) + \left\|A_L\right\|_2\left\|\Sigma_*-A_R\right\|_2
\end{align*}
where in the last line we bound 
\[\|M_*-M\|_2 = \|\widetilde{U}_*^\top U_*-\widetilde{U}_\ell^\top U_k\|_2
\leq \|\widetilde{U}_*^\top-\widetilde{U}_\ell^\top \|_2 + \|U_*-U_k\|_2 \]
and apply Lemma \ref{lem:Vk_Ul}.
\end{proof}

In what follows we make the following assumption needed for stability.

\begin{assumption}[Restricted eigenvalue condition]
    \label{ass:full-column-rank}
    Define $c_{\ell}:=\lambda_{\min }\left(Q_{\ell}^{\top} Q_{\ell}\right)$, and $c_k:=\lambda_{\min }\left(R_k^{\top} R_k\right)$. We assume that $c_\ell, c_k\geq C>0$. This is equivalent to  $Q_{\ell}:=E_{\ell}^{\top} Q$ and $R_k:=F_k^{\top} R$ being full column rank.
\end{assumption}

The constants $c_\ell$ and $c_k$ quantify whether the truncations to $\ell$ instrument directions and $k$ regressor directions retain the singular directions of the overlap operator $M$. Equivalently, $E_\ell^\top$ and $F_k^\top$ act as stable embeddings on the relevant overlap subspaces $\mathrm{span}(Q)$ and $\mathrm{span}(R)$.

Geometrically, the eigenvalues of $Q_{\ell}^\top Q_{\ell}$ are the squared cosines of the principal angles between $\mathrm{span}(Q)$ and the coordinate subspace $\mathrm{span}(E_{\ell})$, so
\[
\lambda_{\min}(Q_{\ell}^\top Q_{\ell})>0
\Longleftrightarrow
\Theta_{\max}\!\left(\mathrm{span}(Q),\mathrm{span}(E_{\ell})\right)<90^\circ,
\]
and analogously for $R_k$.
Thus $c_\ell$ (or $c_k$) close to zero means that at least one overlap direction lies almost entirely in the discarded components, making the induced first-stage operator nearly rank-deficient and inflating the error bounds. In practice, this condition rules out overly aggressive truncation.

Moreover, since $Q$ has orthonormal columns,
\[
\lambda_{\min}(Q_\ell^\top Q_\ell)
=
\lambda_{\min}\!\bigl(Q^\top E_\ell E_\ell^\top Q\bigr)
=
\min_{x\in\mathrm{span}(Q)\setminus\{0\}}
\frac{\|E_\ell^\top x\|_2^2}{\|x\|_2^2},
\]
so a uniform lower bound on $c_\ell$ is exactly a norm-preservation condition for the projection
$E_\ell^\top$ restricted to $\mathrm{span}(Q)$ (and analogously for $c_k$ and $F_k^\top$ on
$\mathrm{span}(R)$). This viewpoint is summarized by the following restricted-eigenvalue condition.

\begin{remark}[Restricted eigenvalue interpretation]
There exist constants $\kappa_\ell,\kappa_k\in(0,1]$ such that
\[
\lambda_{\min}(Q_\ell^\top Q_\ell)\ge \kappa_\ell,
\qquad
\lambda_{\min}(R_k^\top R_k)\ge \kappa_k.
\]
Equivalently, for all $z\in\mathbb R^{\mathrm{rank}(M)}$,
\[
\|Q_\ell z\|_2 \ge \sqrt{\kappa_\ell}\|z\|_2,
\qquad
\|R_k z\|_2 \ge \sqrt{\kappa_k}\|z\|_2.
\]
Since $Q$ and $R$ are isometries, these are equivalent to the restricted norm-preservation bounds
\[
\|E_\ell^\top x\|_2^2 \ge \kappa_\ell \|x\|_2^2 \quad \forall x\in\mathrm{span}(Q),
\qquad
\|F_k^\top x\|_2^2 \ge \kappa_k \|x\|_2^2 \quad \forall x\in\mathrm{span}(R),
\]
i.e., $(E_\ell^\top,F_k^\top)$ act as subspace embeddings on the relevant overlap subspaces.
\end{remark}

\begin{lemma}[Variance bound]
    Under Assumption \ref{ass:full-column-rank}, the variance term can be bounded as
\[\bar\sigma^2 \sum_{i=1}^{\operatorname{rank}(\Delta)} \frac{1}{\sigma_i\left(\Delta\right)^2}\leq\frac{\bar\sigma^2\operatorname{rank}(\Delta)}{c_\ell^2c_k^2\sigma_{\min }(S)^2\sigma_{\min}(A_L)^2\sigma_{\min}(A_R)^2}.\]
    
\end{lemma}

\begin{proof}
The spectrum for $\Delta$ can be bounded as
\begin{align*}
    \sigma_{\min}(\Delta) &= \sigma_{\min}(A_L E_\ell^\top\widetilde{U}^\top U F_kA_R) \\
    &\geq \sigma_{\min}(A_{L})\sigma_{\min}(E_\ell)\sigma_{\min}(\widetilde{U}^\top U)\sigma_{\min}(F_k)\sigma_{\min}( A_{R}) \\
    &= \sigma_{\min}(A_{L})\lambda_{\min }\left(Q_{\ell}^{\top} Q_{\ell}\right)\sigma_{\min}(\widetilde{U}^\top U)\lambda_{\min }\left(R_k^{\top} R_k\right)\sigma_{\min}( A_{R}) \\
    &=c_{\ell} c_k \sigma_{\min }\left(A_L\right) \sigma_{\min }(S)\sigma_{\min }\left(A_R\right)
\end{align*}
 since $Q_{\ell}:=E_{\ell}^{\top} Q$, $R_k:=F_k^{\top} R$ and recall that both $Q$ and $R$ are orthonormal.

 Analogously, 
 \[\sigma_{\max}(\Delta)\leq \sigma_{\max }(S)\sigma_{\max}(A_L)\sigma_{\max}(A_R).\]
 Finally, 
 \[\bar\sigma^2 \sum_{i=1}^{\operatorname{rank}(\Delta)} \frac{1}{\sigma_i\left(\Delta\right)^2}\leq\frac{\bar\sigma^2}{c_\ell^2c_k^2\sigma_{\min}(A_L)^2\sigma_{\min}(A_R)^2}\sum_{i=1}^{\operatorname{rank}(\Delta)}\frac{1}{\sigma_i\left(S\right)^2}\leq\frac{\bar\sigma^2\operatorname{rank}(\Delta)}{c_\ell^2c_k^2\sigma_{\min }(S)^2\sigma_{\min}(A_L)^2\sigma_{\min}(A_R)^2}.\]
\end{proof}

The proof of Theorem \ref{thm:pub} follows immediately after we consolidate all bounds above.

\subsection{Remaining coefficient consistency}

\begin{proof}[Proof of Theorem \ref{thm:2s-ub}]
We write
    \begin{align*}
     &\left\|\left(\Pi^*_{\text{null}}- \Pi_{\text{null}}\right)\right\|_2 + \left\|\left(\Pi^*_{\perp}-\Pi_{\perp}\right)\right\|_2\\
     &= \|V_*V_*^\top-V_kV_k^\top-V_*\mathrm{proj}_{\Delta_*}V_*^\top+V_k\mathrm{proj}_{\Delta}V_k^\top \|_2+ \|V_kV_k^\top-V_*V_*^\top\|_2 \\
     &\lesssim \|V_k\mathrm{proj}_{\Delta_*}V_k^\top-V_k\mathrm{proj}_{\Delta_*}V_k^\top+V_*\mathrm{proj}_{\Delta_*}V_k^\top-V_*\mathrm{proj}_{\Delta_*}V_k^\top-V_*\mathrm{proj}_{\Delta_*}V_*^\top+V_k\mathrm{proj}_{\Delta}V_k^\top \|_2 \\
     & \qquad\qquad + \|V_*V_*^\top-V_kV_k^\top\|_2 \\
     &\leq \|V_k(\mathrm{proj}_{\Delta_*}-\mathrm{proj}_{\Delta})V_k^\top\|_2+\|V_*\mathrm{proj}_{\Delta_*}(V_k^\top-V_*^\top)\|_2 + \|(V_k^\top-V_*^\top)\mathrm{proj}_{\Delta_*}V_k^\top\|_2\\
     & \qquad\qquad + \|V_*V_*^\top-V_kV_k^\top\|_2 \\
     &\leq \|\mathrm{proj}_{\Delta_*}-\mathrm{proj}_{\Delta}\|_2+\|\mathrm{proj}_{\Delta_*}(V_k^\top-V_*^\top)\|_2 + \|(V_k^\top-V_*^\top)\mathrm{proj}_{\Delta_*}\|_2\\
     & \qquad\qquad + \|V_*V_*^\top-V_kV_k^\top\|_2 \\
     &\lesssim \|\mathrm{proj}_{\Delta_*}-\mathrm{proj}_{\Delta}\|_2+\|V_k^\top-V_*^\top\|_2 + \|V_*V_*^\top-V_kV_k^\top\|_2 .
    \end{align*}
    The second and third terms are bounded by Wedin's Theorem as
    \[\|V_k^\top-V_*^\top\|_2 , \|V_*V_*^\top-V_kV_k^\top\|_2 \lesssim \frac{\|H_X\|_2}{\sigma_k(\Sigma_*)}.\]
    To analyze the first term, recall that $\mathrm{proj}_{\Delta_*}$ and $\mathrm{proj}_{\Delta}$ are orthogonal projections onto the row space of $\Delta_* = \widetilde{U}_*^\top U_*\Sigma_*$ and $\Delta = A_L \widetilde{U}_\ell^\top U_kA_R$ respectively. 
    
    Therefore,
    \begin{align*}
        \|\mathrm{proj}_{\Delta_*}-\mathrm{proj}_{\Delta}\|_2 &= \sin\Theta_{\max}\left(\textsc{row}(\Delta_*),\textsc{row}(\Delta)\right) \\
        &= \sin\Theta_{\max}\left(\textsc{row}(\widetilde{U}_*^\top U_*\Sigma_*),\textsc{row}(A_L \widetilde{U}_\ell^\top U_kA_R)\right) \\
        \text{ since left} & \text{ multiplication by an invertible diagonal matrix does not change the row space}: \\
        &= \sin\Theta_{\max}\left(\textsc{row}(\widetilde{\Sigma}_*^{-1}\widetilde{U}_*^\top U_*\Sigma_*),\textsc{row}(\widetilde{\Sigma}_*^{-1}\widetilde{U}_\ell^\top U_kA_R)\right) \\
        &\leq \frac{\|\widetilde{\Sigma}_*^{-1}\widetilde{U}_*^\top U_*\Sigma_*-\widetilde{\Sigma}_*^{-1}\widetilde{U}_\ell^\top U_kA_R\|_2}{\sigma_k(\widetilde{\Sigma}_*^{-1}\widetilde{U}_*^\top U_*\Sigma_*)-\sigma_{k+1}(\widetilde{\Sigma}_*^{-1}\widetilde{U}_*^\top U_*\Sigma_*)} \qquad\qquad\qquad\qquad\text{(Wedin's Theorem)} \\
        \text{ since } & \operatorname{rank}(\widetilde{\Sigma}_*^{-1}\widetilde{U}_*^\top U_*\Sigma_*) = k \text{ by Assumption \ref{ass:full-overlap}}:  \\
        &= \frac{\|\widetilde{\Sigma}_*^{-1}\widetilde{U}_*^\top U_*\Sigma_*-\widetilde{\Sigma}_*^{-1}\widetilde{U}_\ell^\top U_kA_R\|_2}{\sigma_k(\widetilde{\Sigma}_*^{-1}\widetilde{U}_*^\top U_*\Sigma_*)} \\
        \text{ by Courant} & \text{-Fischer } \sigma_k(\widetilde{\Sigma}_*^{-1}\widetilde{U}_*^\top U_*\Sigma_*)\geq \sigma_{\min}(\widetilde{\Sigma}_*^{-1})\sigma_{k}(S_*)\sigma_{\min}(\Sigma_*) = \sigma_{k}(S_*)\frac{\sigma_{\min}(\Sigma_*)}{\sigma_{\max}(\widetilde{\Sigma}_*)}  \text{ , then} \\
        &\leq \frac{\|\widetilde{\Sigma}_*^{-1}\widetilde{U}_*^\top U_*\Sigma_*-\widetilde{\Sigma}_*^{-1}\widetilde{U}_\ell^\top U_kA_R\|_2}{\kappa\left(\widetilde{\Sigma}_*, \Sigma_*\right)^{-1} \sigma_k\left(S_*\right)} \\
         &\leq \frac{\|\widetilde{\Sigma}_*^{-1}\|_2\left(\|\widetilde{U}_*^\top U_*-\widetilde{U}_{\ell}^\top U_k\|_2\|\Sigma_*\|_2+\|\widetilde{U}_\ell^\top U_k(\Sigma_*-A_R)\|_2\right)}{\kappa\left(\widetilde{\Sigma}_*, \Sigma_*\right)^{-1} \sigma_k\left(S_*\right)} \\
        &\leq \frac{\|\widetilde{\Sigma}_*^{-1}\|_2\left(\|\widetilde{U}_*^\top U_*-\widetilde{U}_{\ell}^\top U_k\|_2\|\Sigma_*\|_2+\|\Sigma_*-A_R\|_2\right)}{\kappa\left(\widetilde{\Sigma}_*, \Sigma_*\right)^{-1} \sigma_k\left(S_*\right)} \\
        &\leq \frac{\|\widetilde{\Sigma}_*^{-1}\|_2\left(\left\{\|\widetilde{U}_*-\widetilde{U}_{\ell}\|_2+\|{U}_*-{U}_k\|_2\right\}\|\Sigma_*\|_2+\|\Sigma_*-A_R\|_2\right)}{\kappa\left(\widetilde{\Sigma}_*, \Sigma_*\right)^{-1} \sigma_k\left(S_*\right)} \\
        &\lesssim \frac{\|\widetilde{\Sigma}_*^{-1}\|_2\left(\left\{\frac{\|H_W\|_2}{\sigma_{\ell}(\widetilde{\Sigma}_*)}+\frac{\|H_X\|_2}{\sigma_{k}(\Sigma_*)}\right\}\|\Sigma_*\|_2+\|\Sigma_*-A_R\|_2\right)}{\kappa\left(\widetilde{\Sigma}_*, \Sigma_*\right)^{-1}\sigma_k(S_*)}\qquad\qquad\qquad\text{(Lemma \ref{lem:Vk_Ul})} \\
        &= \frac{\|\widetilde{\Sigma}_*^{-1}\|_2\|\Sigma_*\|_2\frac{\|H_W\|_2}{\sigma_{\ell}(\widetilde{\Sigma}_*)}}{\kappa\left(\widetilde{\Sigma}_*, \Sigma_*\right)^{-1}\sigma_k(S_*)}+\frac{\|\widetilde{\Sigma}_*^{-1}\|_2\|\Sigma_*\|_2\frac{\|H_X\|_2}{\sigma_{k}(\Sigma_*)}}{\kappa\left(\widetilde{\Sigma}_*, \Sigma_*\right)^{-1}\sigma_k(S_*)}+\frac{\|\widetilde{\Sigma}_*^{-1}\|_2\|\Sigma_*-A_R\|_2}{\kappa\left(\widetilde{\Sigma}_*, \Sigma_*\right)^{-1}\sigma_k(S_*)}\\
        &= \frac{\kappa\left(\Sigma_*,\widetilde{\Sigma}_*\right)\frac{\|H_W\|_2}{\sigma_{\ell}(\widetilde{\Sigma}_*)}}{\kappa\left(\widetilde{\Sigma}_*, \Sigma_*\right)^{-1}\sigma_k(S_*)}+\frac{\kappa\left(\Sigma_*,\widetilde{\Sigma}_*\right)\frac{\|H_X\|_2}{\sigma_{k}(\Sigma_*)}}{\kappa\left(\widetilde{\Sigma}_*, \Sigma_*\right)^{-1}\sigma_k(S_*)}+\frac{\sigma_{\ell}(\widetilde{\Sigma}_*)^{-1}\|\Sigma_*-A_R\|_2}{\kappa\left(\widetilde{\Sigma}_*, \Sigma_*\right)^{-1}\sigma_k(S_*)}\\
         &= \frac{\kappa\left(\widetilde{\Sigma}_*\right)\kappa\left(\Sigma_*\right)\|H_W\|_2}{\sigma_k(S_*)\sigma_{\ell}(\widetilde{\Sigma}_*)} + \frac{\kappa\left(\widetilde{\Sigma}_*\right)\kappa\left(\Sigma_*\right)\|H_X\|_2}{\sigma_k(S_*)\sigma_{k}(\Sigma_*)}+\frac{\kappa\left(\widetilde{\Sigma}_*, \Sigma_*\right)\|\Sigma_*-A_R\|_2}{\sigma_{\ell}(\widetilde{\Sigma}_*)\sigma_k(S_*)}.
    \end{align*}

\end{proof}

\begin{corollary}[Theorem~\ref{thm:pub} dominates Theorem~\ref{thm:2s-ub}]\label{cor:projected-dominated}
  The weighting bias from Theorem~\ref{thm:pub} dominates the bound in Theorem~\ref{thm:2s-ub}. In  particular,
  \begin{align*}
      \frac{\left\|I-A_L\right\|_2^2 \sigma_{\max }\left(\Sigma_*\right)^2+\left\|A_L\right\|_2^2\left\|\Sigma_*-A_R\right\|_2^2}{c^2 \sigma_{\min}(S)^2 \sigma_{\min }\left(A_L\right)^2 \sigma_{\min }\left(A_R\right)^2}&\gtrsim\frac{\kappa\left(\widetilde{\Sigma}_*, \Sigma_*\right)^2\|\Sigma_*-A_R\|_2^2}{\sigma_{\ell}(\widetilde{\Sigma}_*)^2\sigma_k(S_*)^2}
  \end{align*}
  and
    \begin{align*}
      \left[\frac{\left\|A_L\right\|_2\sigma_{\max }\left(\Sigma_*\right)}{c \sigma_{\min}(S) \sigma_{\min }\left(A_L\right) \sigma_{\min }\left(A_R\right)}\right]^2 \left(\texttt{NSR}_X+\texttt{NSR}_W\right) &\gtrsim\left[\frac{\kappa\left(\widetilde{\Sigma}_*\right)\kappa\left(\Sigma_*\right)}{\sigma_k(S_*)}\right]^2\left(\texttt{NSR}_X+\texttt{NSR}_W\right).
  \end{align*}
\end{corollary}

\begin{proof}
   Observe that
\begin{align*}
  &\frac{\left\|I-A_L\right\|_2^2 \sigma_{\max }\left(\Sigma_*\right)^2+\left\|A_L\right\|_2^2\left\|\Sigma_*-A_R\right\|_2^2}{c^2 \sigma_{\min}(S)^2 \sigma_{\min }\left(A_L\right)^2 \sigma_{\min }\left(A_R\right)^2}-\frac{\kappa\left(\widetilde{\Sigma}_*, \Sigma_*\right)^2\|\Sigma_*-A_R\|_2^2}{\sigma_{\ell}(\widetilde{\Sigma}_*)^2\sigma_k(S_*)^2} \\
  &\geq  \frac{\left\|A_L\right\|_2^2\left\|\Sigma_*-A_R\right\|_2^2}{c^2 \sigma_{\min}(S)^2 \sigma_{\min }\left(A_L\right)^2 \sigma_{\min }\left(A_R\right)^2}-\frac{\kappa\left(\widetilde{\Sigma}_*, \Sigma_*\right)^2\|\Sigma_*-A_R\|_2^2}{\sigma_{\ell}(\widetilde{\Sigma}_*)^2\sigma_k(S_*)^2} \\
  &= \left[\frac{\kappa(A_L)^2}{c^2\sigma_{\min}(A_R)^2}-\frac{\kappa(\widetilde{\Sigma}_*)^2}{\sigma_{\min}(\Sigma_*)^2}\right] \frac{\|\Sigma_*-A_R\|_2^2}{\sigma_{\min}(S)^2}&\text{(Assumptions \ref{ass:rank}, \ref{ass:full-overlap})}\\
  &\gtrsim 0. &\text{(Assumption \ref{ass:tilting-spectrum})}
\end{align*} 
Moreover, 
 \begin{align*}
      &\left[\frac{\left\|A_L\right\|_2\sigma_{\max }\left(\Sigma_*\right)}{c \sigma_{\min}(S) \sigma_{\min }\left(A_L\right) \sigma_{\min }\left(A_R\right)}\right]^2 \left(\texttt{NSR}_X+\texttt{NSR}_W\right) -\left[\frac{\kappa\left(\widetilde{\Sigma}_*\right)\kappa\left(\Sigma_*\right)}{\sigma_k(S_*)}\right]^2\left(\texttt{NSR}_X+\texttt{NSR}_W\right)\\
      &= \left[\frac{\kappa(A_L)^2}{c^2\sigma_{\min}(A_R)^2}-\frac{\kappa(\widetilde{\Sigma}_*)^2}{\sigma_{\min}(\Sigma_*)^2}\right] \frac{\sigma_{\max }\left(\Sigma_*\right)^2\left(\texttt{NSR}_X+\texttt{NSR}_W\right)}{\sigma_{\min}(S)^2}&\text{(Assumption \ref{ass:full-overlap})}\\
  &\gtrsim 0. &\text{(Assumption \ref{ass:tilting-spectrum})}
  \end{align*}
\end{proof}

\subsection{Comparison of estimators}

\begin{corollary}[Comparing bias terms]
\label{cor:bias-AB-dominance}
Suppose the spectral regime satisfies Assumption \ref{ass:dominance-noise-A}. Then the weighting bias dominates the conditioning bias in Theorem \ref{thm:pub}.
\end{corollary}

\begin{proof}[Corollary \ref{cor:bias-AB-dominance}]
    Denote the noise-driven part of the weighting bias as
$$
\left(A_{\text {noise }}\right)=\frac{\left\|A_L\right\|_2^2\left(\texttt{NSR}_X+\texttt{NSR}_W\right) \sigma_{\max }\left(\Sigma_*\right)^2}{c^2 \sigma_{\min }(S)^2 \sigma_{\min }\left(A_L\right)^2 \sigma_{\min }\left(A_R\right)^2}.
$$
The conditioning bias is
$$
(B)=\texttt{NSR}_X \frac{\kappa(S)^2 \kappa\left(A_L\right)^2 \kappa\left(A_R\right)^2}{c^2} .
$$

Taking the ratio we have that
$$
\begin{aligned}
\frac{B}{A_{\text {noise }}} & =\frac{\texttt{NSR}_X \frac{\sigma_{\max }(S)^2 \sigma_{\max }\left(A_L\right)^2 \sigma_{\max }\left(A_R\right)^2}{c^2 \sigma_{\min }(S)^2 \sigma_{\min }\left(A_L\right)^2 \sigma_{\min }\left(A_R\right)^2}}{\frac{\sigma_{\max }\left(A_L\right)^2\left(\texttt{NSR}_X+\texttt{NSR}_W\right) \sigma_{\max }\left(\Sigma_*\right)^2}{c^2 \sigma_{\min }(S)^2 \sigma_{\min }\left(A_L\right)^2 \sigma_{\min }\left(A_R\right)^2}} \\
& =\frac{\texttt{NSR}_X}{\texttt{NSR}_X+\texttt{NSR}_W} \cdot \frac{\sigma_{\max }(S)^2 \sigma_{\max }\left(A_R\right)^2}{\sigma_{\max }\left(\Sigma_*\right)^2}.
\end{aligned}
$$

Since $S$ encodes the canonical correlations of the overlap $\widetilde{U}^{\top} U$, the condition for dominance of the weighting bias is:
\[\frac{\sigma_{\max }\left(A_R\right)}{\sigma_{\max }\left(\Sigma_*\right)} \leq \sqrt{1+\frac{\texttt{NSR}_W}{\texttt{NSR}_X}}.\]
\end{proof}

\begin{proof}[Proof of Corollary \ref{cor:bias-var-dominated}]
Given Corollary \ref{cor:projected-dominated}, we only consider the weighting bias above.  We write the bias and variance term as
\begin{align*}
(A)&=\frac{\underbrace{\left\|I-A_L\right\|_2^2 \sigma_{\max }\left(\Sigma_*\right)^2+\left\|A_L\right\|_2^2\left\|\Sigma_*-A_R\right\|_2^2}_{\text {weighting bias }}+\underbrace{\left\|A_L\right\|_2^2\left(\texttt{NSR}_X+\texttt{NSR}_W\right) \sigma_{\max }\left(\Sigma_*\right)^2}_{\text {measurement noise}}}{c^2 \sigma_{\min}(S)^2 \sigma_{\min }\left(A_L\right)^2 \sigma_{\min }\left(A_R\right)^2}    \\
(C)&=\frac{\bar{\sigma}^2 r}{c^2 \sigma_{\min}(S)^2 \sigma_{\min}\left(A_L\right)^2 \sigma_{\min}\left(A_R\right)^2}
\end{align*}
with $r=\operatorname{rank}(\Delta)$. Then the variance-dominated region is given by
\[\underbrace{\left\|I-A_L\right\|_2^2 \sigma_{\max }\left(\Sigma_*\right)^2+\left\|A_L\right\|_2^2\left\|\Sigma_*-A_R\right\|_2^2}_{\text {weighting bias }}+\underbrace{\left\|A_L\right\|_2^2 \left(\texttt{NSR}_X+\texttt{NSR}_W\right) \sigma_{\max }\left(\Sigma_*\right)^2}_{\text {measurement noise }} \lesssim \bar{\sigma}^2 r \]
and we simplify to,
\begin{align*}
    \texttt{NSR}:=\texttt{NSR}_X+\texttt{NSR}_W\lesssim \frac{\bar{\sigma}^2 r}{\kappa(X,W)^2}-\frac{\text{weighting bias}}{\kappa(X,W)^2}.
\end{align*}
Note that we are comparing with the $W$-weighting factor since the CCA estimator enters later in the variance-dominated regime, and we can suppose, without loss of generality, $\sigma_{\min}(W)^{-1}\geq\|A_L\|_2$.

Similarly, the following condition is sufficient for a bias-dominated region:
\[\texttt{NSR}:=\texttt{NSR}_X+\texttt{NSR}_W\gtrsim \frac{\bar{\sigma}^2 r}{\kappa(X,W)^2}. \]
\end{proof}

\begin{proof}[Proof of Corollary \ref{cor:estimator-dominance-region}]
    In the variance-dominated region the upper bound ratios between the estimators are:
    \begin{align*}
        \frac{\text{CCA}_{\sqrt{UB}}}{\text{Whitening}_{\sqrt{UB}}} &=\sigma_{\max}(W),\quad 
        \frac{\text{CCA}_{\sqrt{UB}}}{\text{PCA}_{\sqrt{UB}}} =\sigma_{\max}(W)\sigma_{\min}(X),\quad 
        \frac{\text{Whitening}_{\sqrt{UB}}}{\text{PCA}_{\sqrt{UB}}} =\sigma_{\min}(X).
    \end{align*}
Hence, $\text{CCA}$ is better whenever 
\[\sigma_{\max}(W)\leq 1 \quad\text{ and }\quad \sigma_{\max}(W)\sigma_{\min}(X)\leq 1.\]
$\text{PCA}$ is better when
\[1\leq \sigma_{\max}(W)\sigma_{\min}(X)\quad\text{ and }\quad 1 \leq \sigma_{\min}(X).\]
Finally, the whitening is better when
\[1\leq \sigma_{\max}(W)\quad\text{ and }\quad\sigma_{\min}(X)\leq 1.\]

On the other hand, under Assumption \ref{ass:tilting-constant} ignoring the weighting bias, the ratio for comparison in the bias-dominated region is
\begin{align*}
        \frac{\text{CCA}_{\sqrt{UB}}}{\text{Whitening}_{\sqrt{UB}}} &=\kappa(W),\quad 
        \frac{\text{CCA}_{\sqrt{UB}}}{\text{PCA}_{\sqrt{UB}}} =\kappa(W)\sigma_{\min}(X),\quad 
        \frac{\text{Whitening}_{\sqrt{UB}}}{\text{PCA}_{\sqrt{UB}}} =\sigma_{\min}(X)
    \end{align*}
and we get a similar partition from above.

The offset in the phase diagram (gray region) is due to the error term from the weighting bias.
\end{proof}

\subsection{Lower bound}
\label{sec:lower-bound-proof}
We include the statement of Fano's inequality for completeness

\begin{lemma}[Fano's inequality]\label{lem:fano}
Suppose there exists a finite subset $\mathcal{M}$ of the full parameter space $\{\beta_1,\ldots\beta_{|\mathcal{M}|}\}$ such that 
\begin{enumerate}
    \item \textbf{Well separation:} $\|\beta_i-\beta_j\|_2^2\geq d^2$ for all $i\neq j$
    \item \textbf{Statistical indistinguishability:} For any $\beta_i$ and $\beta_j$, there exist corresponding probability distributions $P_i, P_j$ over the observed data that have bounded KL divergence
    $D_{KL}(P_i\|P_j)\leq \alpha.$
\end{enumerate}
Then
\[\inf _{\widehat{\beta}} \sup _{\beta} \mathbb{E}\left[\|\widehat{\beta}-\beta\|_2^2\right] \geq \frac{d^2}{2}\left(1-\frac{\alpha+\log 2}{\log |\mathcal{M}|}\right).\]
\end{lemma}

Below is the proof of the minimax lower bound.

\begin{assumption}[Simplifying rank condition]\label{ass:rank8}
In addition to Assumption \ref{ass:rank}, suppose that $8\leq \operatorname{rank}(\Delta_*) := r$.
\end{assumption}

\begin{proof}[Proof of Theorem \ref{thm:minimax-lb}]
We proceed in steps.
  \begin{enumerate}
  \item \textbf{Data generating process.}
  For a single, fixed set of nonrandom matrices and parameter $(X,W,\beta)$, the probability distribution for the observed data $(Y, \operatorname{vec}(Z_X), \operatorname{vec}(Z_W))\sim P$ is characterized by the independent marginals
  \begin{align*}
      Y &= X\beta^*+\epsilon^*\sim N(X\beta^*,\sigma_\varepsilon^2):=P_Y\\
      \operatorname{vec}(Z_X) &= \operatorname{vec}(X) +\operatorname{vec}(H_X)\sim  N(\operatorname{vec}(X),\sigma_X^2I):=P_{Z_X}\\
     \operatorname{vec}(Z_W) &= \operatorname{vec}(W) +\operatorname{vec}(H_W)\sim  N(\operatorname{vec}(W),\sigma_W^2I):=P_{Z_W}.
  \end{align*}
  Then for any two set of fixed matrices and parameter values $(X,W, \beta)$, $(X^\prime, W^\prime, \beta^\prime)$, we have
  \begin{align*}
      D_{KL}(P\|P^\prime) &= D_{KL}(P_Y\otimes P_{Z_X}\otimes P_{Z_W}\|P^\prime_Y \otimes P^\prime_{Z_X} \otimes P^\prime_{Z_W})\\
      &=D_{KL}(P_Y\|P^\prime_{Y})+D_{KL}(P_{Z_X}\|P^\prime_{Z_X})+D_{KL}(P_{Z_W}\|P^\prime_{Z_W}) 
  \end{align*}
since for independent product measures KL divergence is additive.

  Now we use the KL expression between Gaussians:
\begin{align*}
      D_{\mathrm{KL}}\left(\mathcal{N}(\mu, \Sigma) \| \mathcal{N}\left(\mu^{\prime}, \Sigma\right)\right)
      &=\frac{1}{2}\left(\mu-\mu^{\prime}\right)^{\top} \Sigma^{-1}\left(\mu-\mu^{\prime}\right)\\
      &=\frac{1}{2\sigma_{\varepsilon}^2}\left\|X \beta-X^\prime \beta^\prime\right\|_2^2+\frac{1}{2 \sigma_X^2}\left\|\operatorname{vec}(X)-\operatorname{vec}(X^\prime)\right\|_2^2+\frac{1}{2 \sigma_W^2}\left\|\operatorname{vec}(W)-\operatorname{vec}(W^\prime)\right\|_2^2\\
      &=\frac{1}{2 
     \sigma_{\varepsilon}^2}\left\|X \beta-X^\prime \beta^\prime\right\|_2^2+\frac{1}{2 \sigma_X^2}\left\|X-X^\prime\right\|_{\mathrm{Fr}}^2+\frac{1}{2 \sigma_W^2}\left\|W-W^\prime\right\|_{\mathrm{Fr}}^2.
\end{align*}

      \item \textbf{Construction of the packing.} The previous result shows that the KL divergence is pinned down by $(X,W,\beta)$. Therefore, we will define a sequence of $\left(X^{(\omega)}, W^{(\omega)},\beta^{(\omega)}\right)$, indexed by $(\omega)$, that ultimately defines the sequence of DGPs $(Y^{(\omega)}, Z_X^{(\omega)}, Z_W^{(\omega)})$.\footnote{We drop the vectorized notation for simplicity.}
      
      We start by choosing any two matrices $X^{(0)}, W^{(0)}$ such that $X^{(0)}\in\operatorname{Col}(W^{(0)})$. Let $W^{(0)\top} X^{(0)} = Q_*S_*R_*^{\top}$, and $\sigma_{\min}(W^{(0)})=\sigma_{\min}(\widetilde{\Sigma}_*)$. Denote the $r$-dimensional canonical directions as $\{e_j\}_{j=1}^{r}$. 
      
      By \citet[Lemma 2.9]{tsybakov2009introduction}, there exists a set $\mathcal{M}\subseteq \{0,1\}^{r}$ such that (for $r\geq 8$):
      \begin{enumerate}
          \item $\omega_0 = (0,\ldots, 0)\in \mathcal{M}$,
          \item $\|\omega-\omega^\prime\|_0\geq \frac{r}{8}$, and
          \item $|\mathcal{M}|\geq 2^{\frac{r}{8}}$.
      \end{enumerate}
      This is known in the literature as the Gilbert-Varshamov (GV) pruned hypercube. The construction of the DGPs will be indexed by the elements of the GV hypercube. For each $\omega=(\omega_1,\ldots, \omega_r)\in \mathcal{M}$ define
      \[\beta^{(\omega)} = h\sum_{j=1}^{r}\frac{\omega_j}{\sigma_{j}(S_{*})}R_*e_j\]
      where $\sigma_{\max}(S_{*})\geq \ldots \geq \sigma_{r}(S_{*})>0$  are the cosine of principal angles between $W^{(0)}$ and $X^{(0)}$. Here, $h>0$ is an appropriate scaling to be chosen later. 

    We maintain that the instrumental variables are fixed: $W^{(\omega)} = W^{(0)}$. We now construct $X^{(\omega)}$ as the minimizer of $D_{K L}\left(P^{(\omega)} \| P^{0}\right)$ subject to the restriction that $X^{(\omega)}\in \operatorname{Col}(W^{(0)})$, where $P^{0}$ is the distribution of $(Y^{(0)}, Z_X^{(0)}, Z_W^{(0)})\sim P^{(0)}$, and similarly for $P^{(\omega)}$. In particular, 
\begin{align*}
X^{(\omega)} =& \underset{X\in \operatorname{Col}(W^{(0)})}{\arg\min} D_{K L}\left(P^{(\omega)} \| P^{0}\right)\\
=&\underset{X\in \operatorname{Col}(W^{(0)})}{\arg\min}\frac{1}{2 \sigma_{\varepsilon}^2}\left\|X \beta^{(\omega)}\right\|_2^2+\frac{1}{2 \sigma_X^2}\left\|X-X^{(0)}\right\|_{\mathrm{Fr}}^2\\
=&\underset{X\in \operatorname{Col}(W^{(0)})}{\arg\min}\frac{1}{2 \sigma_{\varepsilon}^2}\left\|\mathrm{proj}_{W^{(0)}}X \beta^{(\omega)}+(I-\mathrm{proj}_{W^{(0)}})X\beta^{(\omega)}\right\|_2^2+\frac{1}{2 \sigma_X^2}\left\|X-X^{(0)}\right\|_{\mathrm{Fr}}^2\\
=&\underset{X\in \operatorname{Col}(W^{(0)})}{\arg\min}\frac{1}{2 \sigma_{\varepsilon}^2}\left\|\mathrm{proj}_{W^{(0)}}X \beta^{(\omega)}\right\|_2^2+\frac{1}{2 \sigma_X^2}\left\|X-X^{(0)}\right\|_{\mathrm{Fr}}^2
\end{align*}
where $\mathrm{proj}_{W^{(0)}}=W^{(0)}(W^{(0)\top} W^{(0)})^\dagger W^{(0)\top}$. Using \citet[Equation 82]{Petersen2008}, the first order condition can be expressed as
\begin{align}
\label{eq:KL-foc}
 X^{(\omega)}-X^{(0)}=-\frac{\sigma_X^2}{\sigma_{\varepsilon}^2}\left(\mathrm{proj}_{W^{(0)}}X^{(\omega)} \beta^{(\omega)}\right) \beta^{(\omega)\top}.   
\end{align}

Instead of solving for the entire matrix $X^{(\omega)}$, we solve for the resulting projected signal vector, $S^{(\omega)}=\mathrm{proj}_{W^{(0)}} X^{(\omega)} \beta^{(\omega)}$, defining it implicitly as
$$
X^{(\omega)}=X^{(0)}-\frac{\sigma_X^2}{\sigma_{\varepsilon}^2}\left(\mathrm{proj}_{W^{(0)}}X^{(\omega)} \beta^{(\omega)}\right) \beta^{(\omega)\top}=X^{(0)}-\frac{\sigma_X^2}{\sigma_{\varepsilon}^2} S^{(\omega)} \beta^{(\omega)\top}.
$$
Left-multiply the entire equation by the projection matrix $\mathrm{proj}_{W^{(0)}}$ and right-multiply by $\beta^{(\omega)}$:
$$
\mathrm{proj}_{W^{(0)}}X^{(\omega)} \beta^{(\omega)}=\mathrm{proj}_{W^{(0)}}X^{(0)} \beta^{(\omega)}-\mathrm{proj}_{W^{(0)}}\left(\frac{\sigma_X^2}{\sigma_{\varepsilon}^2} S^{(\omega)} \beta^{(\omega)\top}\right) \beta^{(\omega)}.
$$
Since $\mathrm{proj}_{W^{(0)}}$ is a projection we simplify as
\[S^{(\omega)}=\mathrm{proj}_{W^{(0)}}X^{(0)} \beta^{(\omega)}-\frac{\sigma_X^2}{\sigma_{\varepsilon}^2}\left(\mathrm{proj}_{W^{(0)}}S^{(\omega)}\right)\left(\beta^{(\omega)\top} \beta^{(\omega)}\right) = \mathrm{proj}_{W^{(0)}}X^{(0)} \beta^{(\omega)}-\frac{\sigma_X^2}{\sigma_{\varepsilon}^2}\|\beta^{(\omega)}\|_2^2S^{(\omega)}.
\]
Therefore, $X^{(\omega)}$ is such that it lies in the subspace defined by $W^{(0)}$ and such that 
$$S^{(\omega)} = \left(1+\frac{\sigma_X^2}{\sigma_{\varepsilon}^2}\|\beta^{(\omega)}\|_2^2\right)^{-1}\mathrm{proj}_{W^{(0)}}X^{(0)} \beta^{(\omega)}.$$ 

Our packing is then characterized by $\left(X^{(\omega)}, W^{(0)},\beta^{(\omega)}\right)_{\omega\in \mathcal{M}}$.

\item \textbf{Well-separation of the parameter.} Note that for any two distinct parameters in the packing,
      \begin{align*}
         \|\beta^{(\omega)}-\beta^{\left(\omega^{\prime}\right)}\|_2^2&=h^2 \left\| \sum_{j=1}^r \frac{\omega_j-\omega_j^{\prime}}{\sigma_{j}(S_{*})} R_* e_j\right\|_2^2 \overset{\text{Parseval}}{=} h^2 \sum_{j=1}^r \frac{(\omega_j-\omega_j^{\prime})^2}{\sigma_{j}(S_{*})^2} = h^2\sum_{j:\omega_j\neq \omega_{j}^\prime}\frac{1}{\sigma_{j}(S_{*})^2}\\
         &\geq \frac{h^2}{\sigma_{\max}(S_{*})^2}\|\omega-\omega'\|_0\geq \frac{h^2r}{8\sigma_{\max}(S_{*})^2}:=d^2 
      \end{align*}
where we are denoting the last quantity as the separation gap $d^2$.

\item \textbf{Bounding the KL divergence.} 
We have that 
\begin{align*}
D_{K L}\left(P^{(\omega)} \| P^{0}\right)&=\frac{1}{2 \sigma_{\varepsilon}^2}\left\|\mathrm{proj}_{W^{(0)}}X^{(\omega)} \beta^{(\omega)}\right\|_2^2+\frac{1}{2 \sigma_X^2}\left\|X^{(\omega)}-X^{(0)}\right\|_{\mathrm{Fr}}^2\\
&=\frac{1}{2 \sigma_{\varepsilon}^2}\left\|S^{(\omega)}\right\|_2^2+\frac{1}{2 \sigma_X^2}\left\|X^{(\omega)}-X^{(0)}\right\|_{\mathrm{Fr}}^2 &\text{(since } S^{(\omega)} = \mathrm{proj}_{W^{(0)}}X^{(\omega)}\beta^{(\omega)})\\
&=\frac{1}{2 \sigma_{\varepsilon}^2}\left\|S^{(\omega)}\right\|_2^2+\frac{1}{2 \sigma_X^2}\left[\left(\frac{\sigma_X^2}{\sigma_{\varepsilon}^2}\right)^2\left\|S^{(\omega)}\beta^{(\omega)\top}\right\|_{\mathrm{Fr}}^2\right] &\text{(Equation \ref{eq:KL-foc})}\\
&=\frac{1}{2 \sigma_{\varepsilon}^2}\left\|S^{(\omega)}\right\|_2^2+\frac{1}{2 \sigma_X^2}\left[\left(\frac{\sigma_X^2}{\sigma_{\varepsilon}^2}\right)^2\left\|S^{(\omega)}\right\|_2^2\left\|\beta^{(\omega)}\right\|_2^2\right] &\text{(Frobenius norm of a rank-1 matrix)}\\
&=\frac{1}{2\sigma_{\varepsilon}^2}\left\|S^{(\omega)}\right\|_2^2\left(1+\frac{\sigma_X^2}{\sigma_{\varepsilon}^2}\|\beta^{(\omega)}\|_2^2\right).
\end{align*}
Now substituting back $S^{(\omega)}= \left(1+\frac{\sigma_X^2}{\sigma_{\varepsilon}^2}\|\beta^{(\omega)}\|_2^2\right)^{-1}\mathrm{proj}_{W^{(0)}}X^{(0)} \beta^{(\omega)}$ yields
\[D_{K L}\left(P^{(\omega)} \| P^{0}\right) =\frac{1}{2} \cdot \frac{\left\|\mathrm{proj}_{W^{(0)}}X^{(0)} \beta^{(\omega)}\right\|_2^2}{\sigma_{\varepsilon}^2+\sigma_X^2\left\|\beta^{(\omega)}\right\|_2^2}.\]
Note that
\begin{align*}
 W^{(0)\top} X^{(0)} \beta^{(\omega)}&=W^{(0)\top} X^{(0)}\left(h \sum_{j=1}^r \frac{\omega_j}{\sigma_{j}(S_{*})} R_* e_j\right)=h \sum_{j=1}^r \frac{\omega_j}{\sigma_{j}(S_{*})} W^{(0)\top} X^{(0)}R_* e_j\\
 &=h \sum_{j=1}^r \frac{\omega_j}{\sigma_{j}(S_{*})} \sigma_{j}(S_{*}) Q_*e_j =h \sum_{j=1}^r \omega_j Q_*e_j.
\end{align*}
Therefore,
\begin{align*}
\left\|\mathrm{proj}_{W^{(0)}}X^{(0)} \beta^{(\omega)}\right\|_2^2 =&\left\|W^{(0)}(W^{(0)\top} W^{(0)})^{\dagger}W^{(0)\top} X^{(0)} \beta^{(\omega)}\right\|_2^2\leq \left\|W^{(0)}(W^{(0)\top} W^{(0)})^{\dagger}\right\|_2^2\left\|h \sum_{j=1}^r \omega_j q_{*j}\right\|_2^2\\
\overset{\text{Parseval}}{=}&\frac{1}{\sigma_{\min}(\widetilde{\Sigma}_*)^2}h^2\|\omega\|_2^2=\frac{1}{\sigma_{\min}(\widetilde{\Sigma}_*)^2}h^2\|\omega\|_0\leq\frac{h^2r}{\sigma_{\min}(\widetilde{\Sigma}_*)^2}.
\end{align*}

Consolidating, together with $\left\|\beta^{(\omega)}\right\|_2^2\geq \frac{h^2r}{8\sigma_{\max}(S_{*})^{2}}$, we obtain an upper bound for the KL divergence between any two elements in the packing
\[D_{K L}\left(P^{(\omega)} \| P^{(\omega^\prime)}\right)\leq 2 D_{K L}\left(P^{(\omega)} \| P^0\right)+2D_{K L}\left(P^0 \| P^{(\omega^\prime)}\right)\leq   2\cdot \frac{\frac{h^2r}{\sigma_{\min}(\widetilde{\Sigma}_*)^2}}{\sigma_{\varepsilon}^2+\sigma_X^2\left(\frac{h^2r}{8\sigma_{\max}(S_{*})^{2}}\right)}:=\alpha.\]

As an aside, observe that the KL divergence does not satisfy the triangle inequality. However, restricting to normal distributions with same covariance, the square-root of the KL divergence is indeed a metric. For the family $N(\mu, \Sigma)$ with fixed $\Sigma$, one has
$$
D_{KL}\left(N\left(\mu_1, \Sigma\right) \| N\left(\mu_2, \Sigma\right)\right)=\frac{1}{2}\left\|\mu_1-\mu_2\right\|_{\Sigma^{-1}}^2,
$$
so $\sqrt{2 D_{K L}}=\left\|\mu_1-\mu_2\right\|_{\Sigma^{-1}}$ is a (Mahalanobis) distance and satisfies the triangle inequality.

\item \textbf{Application of Fano's Inequality (Lemma \ref{lem:fano}).}
In order to obtain a meaningful bound we require that 
\[\frac{r}{8}\log 2 = \log |\mathcal{M}| \gtrsim \alpha := 2 \frac{\frac{h^2r}{\sigma_{\min}(\widetilde{\Sigma}_*)^2}}{\sigma_{\varepsilon}^2+\sigma_X^2\left(\frac{h^2 r}{8\sigma_{\max}(S_{*})^{2}}\right)}.\]
We choose the scaling $h$ such that it satisfies the inequality above. Note that such  $h$ satisfies
\[h^2\asymp \max\left\{\frac{\sigma_{\max}(S_{*})^2\sigma_\varepsilon^2\sigma_{\min}(\widetilde{\Sigma}_*)^2}{\sigma_{\max}(S_{*})^2-\frac{r \sigma_X^2\sigma_{\min}(\widetilde{\Sigma}_*)^2}{8}},\sigma_\varepsilon^2\sigma_{\min}(\widetilde{\Sigma}_*)^2\right\}.\]

Therefore we can conclude that
\begin{align*}
   \inf _{\widehat{\beta}} \sup _{\beta} \mathbb{E}\left[\|\widehat{\beta}-\beta\|_2^2\right] &\geq \frac{d^2}{2}\left(1-\frac{\alpha+\log 2}{\log |\mathcal{M}|}\right) \\
   &\gtrsim \max\left\{\frac{\sigma_{\max}(S_{*})^2\sigma_\varepsilon^2\sigma_{\min}(\widetilde{\Sigma}_*)^2}{\sigma_{\max}(S_{*})^2-\frac{r \sigma_X^2\sigma_{\min}(\widetilde{\Sigma}_*)^2}{8}},\sigma_\varepsilon^2\sigma_{\min}(\widetilde{\Sigma}_*)^2\right\}\frac{r}{\sigma_{\max}(S_{*})^2} \\
   &= \max\left\{\frac{\sigma_\varepsilon^2 r \sigma_{\min}(\widetilde{\Sigma}_*)^2}{\sigma_{\max}(S_{*})^2-\frac{r \sigma_X^2\sigma_{\min}(\widetilde{\Sigma}_*)^2}{8}},\frac{\sigma_\varepsilon^2 r \sigma_{\min}(\widetilde{\Sigma}_*)^2}{\sigma_{\max}(S_{*})^2}\right\}.
\end{align*}
  \end{enumerate}  
\end{proof}

\begin{remark}[Relaxing Assumption~\ref{ass:rank8}]
    The constraint $r>8$ (Assumption \ref{ass:rank8}) is just a convenience tied to the Gilbert-Varshamov packing;  it is not fundamental. We can remove it by a more careful construction of the GV hypercube:
      \begin{enumerate}
          \item $\omega_0 = (0,\ldots, 0)\in \mathcal{M}$,
          \item $\|\omega-\omega^\prime\|_0\geq \lceil\frac{r}{8}\rceil$, and
          \item $\log|\mathcal{M}|\geq cr$ for a universal constant $c>0$.
      \end{enumerate}
\end{remark}

\section{Simulation details}\label{sec:details}

\subsection{Concentrated signal}
For each replication we generate a rank-$k$ latent regressor matrix
\[
X = U_X \Sigma_X V_X^\top \in \mathbb{R}^{n\times p},
\qquad
\Sigma_X = \mathrm{diag}\bigl((i+1)^{-\alpha}\bigr)_{i=1}^k,
\]
where $U_X\in\mathbb{R}^{n\times k}$ and $V_X\in\mathbb{R}^{p\times k}$ have orthonormal columns
(drawn via QR of i.i.d.\ Gaussian matrices).
We generate a rank-$\ell$ latent instrument matrix
\[
W = U_W \Sigma_W V_W^\top \in \mathbb{R}^{n\times p_w},
\qquad
\Sigma_W = \mathrm{diag}\bigl((i+1)^{-\alpha}\bigr)_{i=1}^\ell,
\]
with $V_W\in\mathbb{R}^{p_w\times \ell}$ orthonormal and with $U_W$ constructed to have controlled
overlap with the signal subspace of $X$.
Specifically, letting $r=\min\{k,\ell\}$, we form the first $r$ columns of $U_W$ as
\[
(U_W)_{:,1:r}
=
\delta(U_X)_{:,1:r}
+
\sqrt{1-\delta^2}(U_\perp)_{:,1:r},
\]
where $U_\perp$ spans an orthogonal complement of $\mathrm{span}(U_X)$, and then orthonormalize.
The parameter $\delta\in[0,1]$ controls the strength of instrument alignment: larger $\delta$ yields
larger canonical correlations between the signal subspaces of $X$ and $W$.

\subsection{Diffuse noise}

We observe noisy measurements
\[
Z_X = X + H_X, \qquad Z_W = W + H_W,
\]
where $(H_X,H_W)$ are Gaussian measurement errors with a shared component.
Concretely, on the first $p_{\mathrm{common}}=\min\{p,p_w\}$ columns we draw
$(H_{Wi,j},H_{Xi,j})$ jointly normal with correlation $\rho$ and marginal standard deviations
$\sigma_{H_W},\sigma_{H_X}$, and draw the remaining columns independently.
We scale
\[
\sigma_{H_X} = \frac{c_1}{(k+1)^\alpha\sqrt{p}},
\qquad
\sigma_{H_W} = \frac{c_1}{(\ell+1)^\alpha\sqrt{p_w}},
\]
so that the magnitude of corruption is comparable across dimensional regimes.

\subsection{Instrument validity}

We generate a disturbance that is correlated with $X$ but orthogonal to the instrument space:
\[
\varepsilon_0 = X\gamma + \eta,
\qquad
\varepsilon = (I-\mathrm{proj}_W)\varepsilon_0,
\qquad
Y = X\beta^* + \varepsilon,
\]
where $\mathrm{proj}_W$ is the orthogonal projector onto $\mathrm{span}(U_W)$,
$\eta\sim \mathcal{N}(0,\sigma_\varepsilon^2 I_n)$, and $\gamma\in\mathbb{R}^p$ is fixed within each
$(n,\delta,\text{regime})$ block.
By construction $\mathbb{E}[\varepsilon\mid W]=0$ while allowing endogeneity
$\mathbb{E}[\varepsilon\mid X]\neq 0$.
We rescale $\varepsilon$ to have variance $\sigma_\varepsilon^2$ for comparability across designs.

\subsection{Dimension}
We consider two regimes. In the moderate-dimensional regime we set $p=\lfloor n/2\rfloor$ and $p_w=\lfloor n/3\rfloor$.
In the  high-dimensional regime we set $p=n-100$ and $p_w=n-200$ (capped at $5000$).
We fix $(k,\ell)=(8,10)$, $\alpha=1.5$, $\rho=0.9$, $\sigma_\varepsilon=1.25$, and $c_1=2.0$.
Instrument alignment is varied over $\delta\in\{0.001,0.05,0.65\}$, and we report results across
$n\in\{300,500,1000,2000,5000\}$.

\end{document}